\documentclass[a4paper,10pt,USenglish,cleveref,capitalise,autoref,thm-restate,numberwithinsect,nameinlink]{article}
\usepackage[margin=1.3in]{geometry}

\usepackage{multicol}		% used for the two-column index

\usepackage{lineno}
\usepackage{xspace}

\usepackage{amssymb}
\usepackage{amsmath}
\usepackage{amsthm}
\usepackage{enumerate}

\usepackage{tcolorbox}
\usepackage{dsfont}
\usepackage{thmtools}
\usepackage{thm-restate}
\usepackage[edges]{forest}
\usepackage{multirow}
\usepackage{color}
\usepackage{accents}
\usepackage{tikz}
\usepackage{todonotes}
\usepackage{comment}
\usepackage{placeins}
\usepackage{nicefrac}
\usepackage[USenglish]{babel}
\usepackage[]{algorithm2e}
\usepackage[colorlinks = true,
			linkcolor = blue,
			urlcolor  = blue,
			citecolor = blue,
			anchorcolor = blue]{hyperref}

\newif\iflong\longfalse

\newcommand\Problem[3]{
\begin{center}
\begin{minipage}[c]{.95\linewidth}
\textsc{#1}\\
\textbf{Input}: #2.\\
\textbf{Question}: #3?
\end{minipage}
\end{center}}
\newcommand{\pfpara}[1]{\smallskip\noindent\textit{#1.}}

\newtheorem{theorem}{Theorem}
\newtheorem{corollary}{Corollary}
\newtheorem{proposition}{Proposition}

\newcommand{\mcal}{\mathcal}
\newcommand{\mbb}{\mathbb}
\newcommand{\Instance}{\ensuremath{\mathcal{I}}\xspace}
\newcommand{\Oh}{\ensuremath{\mathcal{O}}\xspace}
\newcommand{\W}[1]{\ensuremath{\mathrm{W[#1]}}\xspace}
\newcommand\NP{\ensuremath{\mathrm{NP}}\xspace}
\newcommand\paraNP{\ensuremath{\mathrm{para}}-\ensuremath{\mathrm{NP}}\xspace}
\newcommand\FPT{\ensuremath{\mathrm{FPT}}\xspace}
\newcommand\XP{\ensuremath{\mathrm{XP}}\xspace}

\DeclareMathOperator{\wcode}{w-code}
\DeclareMathOperator{\height}{height}
\DeclareMathOperator{\val}{val}
\DeclareMathOperator{\bin}{bin}
\DeclareMathOperator{\var}{var}

\DeclareMathOperator{\off}{off}
\DeclareMathOperator{\Costs}{Cost}
\newcommand{\numP}{\left\lVert \mathcal{P}\right\rVert}

\newcommand\Tree{\ensuremath{\mathcal{T}}\xspace}
\newcommand{\bet}[1]{\ensuremath{\left|#1\right|}}
\newcommand{\sprob}{sur\-vi\-val pro\-ba\-bil\-ity\xspace}
\newcommand{\sprobs}{sur\-vi\-val pro\-ba\-bi\-li\-ties\xspace}
\newcommand{\myvec}[1]{\ensuremath{\vec {#1}}}
\newcommand{\mo}{-_{\ge 0}}

\newcommand\KP{\textsc{Knapsack}\xspace}
\newcommand\SubSum{\textsc{Subset Sum}\xspace}

\newcommand\ILPF{\textsc{ILP-Feasibility}\xspace}
\newcommand\MCKP{\textsc{MCKP}\xspace}
\newcommand\MCKPLong{\textsc{Multiple-Choice Knapsack}\xspace}
\newcommand\PS{\textsc{Pe\-nal\-ty-Sum}\xspace}
\newcommand\GNAP{\textsc{GNAP}\xspace}
\newcommand\SGNAP{\textsc{Star-GNAP}\xspace}
\newcommand\GNAPLong{\textsc{Generalized Noah's Ark Problem}\xspace}
\newcommand\NAPold[3][]{\ensuremath{#1 \stackrel{#2}{\to} #3}~\textsc{NAP}\xspace}
\newcommand\NAP[4][]{\NAPold[\ensuremath{#1}]{\ensuremath{#2}}{\ensuremath{#3}}}

\newcommand\unitc{\textsc{Unit-Cost NAP}\xspace}

\usepackage{soul}

\newcommand{\kommentar}[1]{}

\begin{document}
  
\title{A Multivariate Complexity Analysis of the Generalized Noah's Ark Problem}

\author{Christian Komusiewicz%
	\thanks{\href{https://orcid.org/0000-0003-0829-7032}{ORCID: 0000-0003-0829-7032}.}\\
		Friedrich Schiller University Jena, Jena, Germany\\
		c.komusiewicz@uni-jena.de
	\and
	Jannik Schestag%
	\thanks{\href{https://orcid.org/0000-0001-7767-2970}{ORCID: 0000-0001-7767-2970}. Partially funded by the Dutch Research Council (NWO), project “Optimization for and with Machine Learning (OPTIMAL)” OCENW.GROOT.2019.015.
	The research was partially carried out during an extended research visit of Jannik Schestag at TU Delft, The Netherlands, in 2023. We thank the German Academic Exchange Service (DAAD), project no. 57556279, for the funding.}\\
		Delft University of Technology, Delft, The Netherlands\\
		j.t.schestag@tudelft.nl
}

\date{}

\maketitle

% \institute{Christian Komusiewicz \href{https://orcid.org/0000-0003-0829-7032}{\orcidID{0000-0003-0829-7032}},
% \email{c.komusiewicz@uni-jena.de},
% \at Jannik Schestag \href{https://orcid.org/0000-0001-7767-2970}{\orcidID{0000-0001-7767-2970}}, \email{j.t.schestag@uni-jena.de},
% \at Friedrich-Schiller-Universität Jena,  Ernst-Abbe-Platz 2, 07743 Jena, Germany;
% %\at Parts of the research were made on a research exchange of Jannik Schestag at the TU Delft. The trip was funded by the German Academic Exchange Service (DAAD), programme no. 57556279.
% \at


	\begin{abstract}
		In the \GNAPLong, one is given a phylogenetic tree on a set of species~$X$ and a set of conservation projects for each species.
		Each project comes with a cost and raises the survival probability of the corresponding species.
		The aim is to select a conservation project for each species such that the total cost of the selected projects does not exceed some given threshold and the expected phylogenetic diversity is as large as possible.
		We study the complexity of \GNAPLong and some of its special cases with respect to several parameters related to the input structure, such as the number of different costs, the number of different survival probabilities, or the number of species,~$\bet X$.
	\end{abstract}

\section{Introduction}
The preservation of biological diversity is one of humanity's most critical challenges. To help systematically address this challenge, it is useful to quantify or predict the effect of interventions. Here, two questions arise: how to measure the biological diversity of ecosystems and how to model the effect of certain actions on the biological diversity of an ecosystem under consideration. 

A popular way to measure the biological diversity of an ecosystem, introduced by Faith~\cite{faith92}, is to consider the \emph{phylogenetic diversity} of the species present in that system. Here, the phylogenetic diversity is the sum of evolutionary distances between the species when their evolution is modeled by an evolutionary (phylogenetic) tree. The tree then not only gives the phylogenetic diversity of the whole species set but also allows us to infer the phylogenetic diversity of any subset of these species that would remain after some currently present species become extinct. Now, to model the effect of certain actions, a first simple model is that one can afford to protect $k$~species and that all other species go extinct. Maximizing phylogenetic diversity under this model can be solved in polynomial time with a simple greedy algorithm~\cite{steel,pardi05}. Later, more realistic models were introduced. One step was to model that protecting some species may be more costly than protecting others~\cite{pardi}. Subsequent approaches also allowed us to model uncertainty as follows: performing an action to protect some species does not guarantee the survival of that species but only raises the survival probability~\cite{weitzman}. In this model, one now aims to maximize the \emph{expected} phylogenetic diversity. Finally, one may also consider the even more realistic case when, for each species, one may choose from a set of different actions or even from combinations of different actions. Each choice is then associated with a cost and with a resulting survival probability. This model was proposed by Pardi~\cite{pardidiss} and~Billionnet~\cite{billionnet13,billionnet17} as~\GNAPLong{} (\GNAP).

Introducing cost differences for species protection makes the problem of maximizing phylogenetic diversity NP-hard~\cite{pardi} and thus all of the even richer models are NP-hard as well. Apart from the NP-hardness and several pseudopolynomial-time algorithms, there is no work that systematically studies which structural properties of the input make the problems tractable. This work aims to fill this gap. More precisely, we study how different parameters related to the input structure influence the algorithmic complexity of \GNAP{} and several of its special cases.

In a nutshell, we show the following. First, \GNAP{} can be solved efficiently when the number of different project costs and survival probabilities is small. Second, while a constant number~$\bet X$ of species~$X$ leads to polynomial-time solvability, algorithms with running time~$f(|X|)\cdot | \mathcal{I}| ^{\Oh(1)}$ are unlikely to exist. Finally, restricted cases where, for example the input tree has height~1 or there is exactly one action for each species that guarantees its survival are much easier than the general problem. Some of our results are obtained by observing a close relation to the \MCKPLong~problem, for which we also  provide a number of new complexity results. 

This work is organized as follows. In Section~\ref{sec:prelims}, we give formal definitions of the problem and the considered parameters and an overview of our results. In Section~\ref{sec:MCKP}, we study \MCKPLong. In Section~\ref{sec:GNAP}, we consider \GNAP{}, first the general case and then the case where the phylogenetic tree has height~1. In Section~\ref{sec:two-projects}, we consider the special case of \GNAP{} when there are at most two possible actions for each species. 

\section{Basic Definitions and Overview of the Results}
\label{sec:prelims}%\paragraph{Mathematical Definitions}
For an integer~$n$, by~$[n]$ we denote the set~$\{1,\dots,n\}$, and by~$[n]_0$ we denote the set~$\{0,1,\dots,n\}$.
We generalize functions $f: A\to B$ to handle subsets~$A' \subseteq A$ by~$f(A') := \bigcup_{a\in A'} f(a)$ if~$B$ is a family of sets and~$f_\Sigma(A') := \sum_{a\in A'} f(a)$ if~$B \subseteq \mathbb{R}$.
A \textit{partition of~$N$} is a family of pairwise disjoint sets~$\{ N_1, \dots, N_m \}$ such that~$\bigcup_{i=1}^m N_i=N$.

\subsection{Parameterized Complexity}
To assess the influence of structural properties of the input on the problem complexity, we study the problems in the framework of parameterized complexity.
For a detailed introduction to parameterized complexity refer to the standard monographs~\cite{cygan,downeybook}. We only  recall the most important definitions: Instances~$(I,k)$ of a parameterized problem consist of a classical problem instance~$I$ and a parameter~$k$. A parameterized problem with parameter~$k$ is \emph{fixed-parameter tractable} (\FPT) if every instance~$(\mathcal I,k)$ can be solved in~$f(k)\cdot |\mathcal{I}|^{\Oh(1)}$ time.
If every instance can be solved with an algorithm such that~$f$ is a polynomial, then we say the parameterized problem is \emph{polynomial fixed-parameter tractable} (PFPT).
A parameterized problem is \emph{slicewice-polynomial}~(\XP) if every instance can be solved in~$|\mathcal{I}|^{g(k)}$ time. Some problems in \XP are assumed to have no \FPT-algorithm. In particular, there is the complexity class \W 1 for which $\FPT\subseteq \W 1 \subseteq \XP$ is known and it is widely assumed that problems that are \W 1-hard have no \FPT-algorithm. To prove that a problem is \W 1-hard, one uses parameterized reductions. A \emph{parameterized reduction} from a parameterized problem~$L$ to a parameterized problem~$L'$ is an algorithm which, given an instance $(\mathcal I,k)$ of~$L$, constructs in~$f(k)\cdot |\mathcal{I}|^{\Oh(1)}$~time an equivalent instance~$(I',k')$ of~$L'$ such that~$k'\le g(k)$ for some function~$g$. A parameterized problem~$L$ is called \emph{\paraNP-hard} if there exists some constant value~$t$ such that~$L$ restricted to instances with~$k\le t$ is NP-hard. 

In our running time analyses, we assume a unit-cost RAM model where arithmetic operations have constant running time.

\subsection{Trees}
A \textit{directed graph}~$G$ is a tuple~$(V,E)$, where~$V$ is called the \textit{set of vertices} of~$G$ and $E$ the \textit{set of edges} of~$G$, respectively. We denote an edge directed from~$u$ to~$v$ by~$(u,v)$.
The \textit{degree of a vertex~$v$} is the number of edges that are incident with~$v$.
A \textit{tree~$T$ with root~$r$} is a directed acyclic graph with a distinguished vertex~$r$ such that each vertex of~$T$ can be reached from~$r$ via exactly one path.
A vertex~$v$ of a tree~$T$ is a \textit{leaf} when the degree of~$v$ is one.
The~\textit{height of a rooted tree~$T$}, denoted ~$\height_T$, is the maximal distance of its root~$r$ to any leaf. 
For an edge~$(u,v)$, we call~$u$~\textit{the parent of~$v$} and~$v$~\textit{a child of~$u$}.
For a vertex~$v$ with parent $u$, the \textit{subtree $T_v$ rooted at $v$} is the connected component containing $v$ in~$T-(u,v)$. When $v$ is the root of $T$, we define $T_v:=T$. For a vertex~$v$ in a tree~$T$, the~\textit{offspring of~$v$}, denoted~$\off(v)$,~is the set of leaves in~$T_v$.
For any vertex~$v$, we assume an arbitrary but fixed ordering of its children.
A~\textit{star} is a graph~$G=(V,E)$ with a vertex~$c\in V$ and~$E=\{ (c,v) \mid v\in V\setminus\{c\} \}$.

\subsection{Phylogenetic Diversity}
A phylogenetic \textit{$X$-tree~$\Tree=(V,E,\lambda)$} (in short,~$X$-tree) is a tree~$T$ with root~$r$, where~$\lambda: E\to \mathbb{N}$ is an edge-weight function and~$X$ is the set of leaves of~$T$. In biological applications, the internal vertices of~$T$ correspond to hypothetical ancestors of the leaves and~$\lambda(e)$ describes the evolutionary distance between the endpoints of~$e$.
An~$X$-tree~$\Tree$ is \textit{ultrametric} if there is an integer~$p$ such that the sum of the weights of the edges from the root~$r$ to~$x_i$ equals~$p$ for every leaf~$x_i$.
% For an~$X$-tree~$\Tree=(V,E,\lambda)$, a \textit{subtree with root~$v$}, denoted~$\Tree_v$ is the rooted tree that contains the offspring of~$v$ and no further vertices.

For a taxon~$x_i\in X$, a \textit{project list}~$P_i$ is a tuple~$(p_{i,1},\dots,p_{i,\ell_i})$. Herein, each project~$p_{i,j}$ represents an action that may raise the survival probability of~$x_i$ at a certain cost.  Formally, each \emph{project} is a tuple~$(c_{i,j},w_{i,j})\in \mbb N_0\times \mbb Q\cap [0,1]$, where~$c_{i,j}$ is the \textit{cost} and~$w_{i,j}$ the \textit{\sprob} of~$p_{i,j}$.
 As a project with higher cost will only be considered when the \sprob is higher, we assume the costs and \sprobs{} to be ordered. That is,~$c_{i,j}<c_{i,j+1}$ and~$w_{i,j}<w_{i,j+1}$ for every project list~$P_i$ and every~$j<\ell_i$.
An \textit{$m$-collection of project lists}~$\mcal P$ is a set of~$m$~project lists~$\{P_1,\dots,P_m\}$.
For a project set~$S$, the~\textit{total cost~$\Costs(S)$ of~$S$} is~$\sum_{p_{i,j}\in S} c_{i,j}$.

For a given~$X$-tree~$\Tree$, the~\textit{expected phylogenetic diversity~$PD_\Tree(S)$ of a set of projects~$S=\{p_{1,j_1}, \dots, p_{|X|,j_{|X|}}\}$} is given by

$$
PD_\Tree(S) := \sum_{(u,v)\in E} \lambda(u,v) \cdot \left(1 - \prod_{x_i\in \off(v)} (1 - w_{i,j_i}) \right).
$$
Herein, the term~$\left(1 - \prod_{x_i\in \off(v)} (1 - w_{i,j_i}) \right)$ is the probability that at least one offspring of~$v$ (and with it the edge~$(u,v)$) survives  when~$S$ is implemented.  % Thus, the phylogenetic diversity is the sum of the expected values of the edges of~$\Tree$ when applying~$S$.

\subsection{Problem Definitions, Parameters, and Results Overview}
%\todos[inline]{Diskussion: Ich würde gerne bei~$\var_w$ bleiben und Wahrscheinlichkeiten mit~$w_i$ ausdrücken, da~$p$ bereits für projects reserviert ist. Billionnet hat auch für die Wahrscheinlichkeiten~$w$ benutzt~\cite{billionnet13}.}
The most general version of the problem under consideration is defined as follows.
\Problem{\GNAPLong (\GNAP)}
{An~$X$-tree~$\Tree=(V,E,\lambda)$,
an~$\bet{X}$-collection of project lists~$\mathcal{P}$,
and numbers~$B\in \mbb N_0,D\in \mbb Q_{\ge 0}$}
{Is there a set of projects~$S=\{p_{1,j_1}, \dots, p_{|X|,j_{|X|}}\}$, one from each project list of~$\mathcal{P}$, 
such that~$PD_\Tree(S)\ge D$ and~$\Costs(S)\le B$}
Such a project set~$S$ is called~\textit{a solution for the instance~$\mathcal I=(\Tree,\mathcal P,B,D)$}. % We refer the special case of~\GNAP where~$\Tree$ is restricted to have height 1 as \SGNAP. The special case of \GNAP where each species has two projects is defined as follows.

We consider the following special cases of \GNAP:
\begin{itemize}
\item \SGNAP is the special case of~\GNAP where~$\Tree$ is restricted to have height 1.
\item \NAP[a_i]{c_i}{b_i}{2} is the special case of~\GNAP where each project list~$P_i$ contains exactly two projects~$(0,a_i)$ and~$(c_{i},b_i)$.
  In other words, in an instance of~\NAP[a_i]{c_i}{b_i}{2} we can decide for each taxon~$x_i$ whether we want to spend~$c_i$ to increase the \sprob of~$x_i$ from~$a_i$ to~$b_i$. If the \sprob of all species is either 0 or some uniform value, then this is denoted by~\NAP[0]{c_i}{b}{2}. For~$b=1$ this corresponds to the choice between guaranteed extinction and guaranteed survival.  
% Let us remark that the case with two projects with strictly positive cost can be easily reduced to \NAP[a_i]{c_i}{b_i}{2}: If $P_i$ contains two projects~$(\bar c,a_i)$ and~$(\hat c,b_i)$ with $0<\bar c <  \hat c$ for some~$i\in [\bet{X}]$, then reducing the cost of both projects and the total budget by $\bar c$ gives an equivalent instance.
\item \unitc is the special case of \NAP[a_i]{c_i}{b_i}{2} where every project with a positive \sprob has cost~1.
\end{itemize}

%\paragraph{Parameters}
We study \GNAP with respect to several parameters which we describe in the following. The results are summarized in Table~\ref{tab:results-gnap}  for the general problem and in Table~\ref{tab:results-special gnap} for the special case with two projects per taxon. % If not stated differently, we assume in the following that~$i \in [\bet X]$ and~$j\in [|P_i|]$.
The input of \GNAP directly gives the natural parameters \textit{number of taxa~$|X|$}, \textit{budget~$B$}, and required \textit{diversity~$D$}. Closely related to~$B$ is the \textit{maximum cost per project~$C=\max_{i,j} c_{i,j}$}. We may assume that no projects have a cost that exceeds the budget, as we can delete them from the input and so~$C\le B$. We may further assume that~$B\le C\cdot \bet X$, as otherwise we can compute in polynomial time whether the diversity of the most valuable projects of the taxa exceeds~$D$ and return yes, if it does and no, otherwise.
\begin{table}[t]
	\centering
%	\footnotesize
	\caption{Complexity results for \GNAP with unbounded number of projects per taxon.}
	\label{tab:results-gnap}
	\begin{tabular}{lll}
		\hline
		Parameter & \GNAP & \SGNAP \\
		\hline
		$|X|$ & \W1-hard (Thm.~\ref{thm:X-W1hard}), XP (Prop.~\ref{pps:BruteForce}\ref{pps:GNAP-X-XP}) & \W1-hard (Thm.~\ref{thm:X-W1hard}), XP (Prop.~\ref{pps:BruteForce}\ref{pps:GNAP-X-XP})\\
		$B$ & XP (Prop.~\ref{pps:BruteForce}\ref{obs:XP-B}) & PFPT $\Oh(B \cdot \numP)$ (Prop.~\ref{pps:height=1->mckp})\\
		$C$ & para-NP-h (Thm.~\ref{thm:PS-hardness}) & PFPT $\Oh(C \cdot \numP \cdot |X|)$ (Prop.~\ref{pps:height=1->mckp})\\
		$D$ & para-NP-h (Prop.~\ref{pps:stars-obs}\ref{obs:D=height=1,varw=2}) & para-NP-h (Prop.~\ref{pps:stars-obs}\ref{obs:D=height=1,varw=2})\\
		\hline
		$\var_c+\var_w$ & XP $\Oh(|X|^{2\cdot (\var_c + \var_w) +1})$ (Thm.~\ref{thm:varc+varw}) & FPT (Thm.~\ref{thm:height=1:varc+varw})\\
		$\var_c$ & para-NP-h (Thm.~\ref{thm:PS-hardness}) & XP $\Oh(|X|^{\var_c-1} \cdot \numP)$ (Prop.~\ref{pps:height=1->mckp})\\
		$\var_w$ & para-NP-h (Prop.~\ref{pps:stars-obs}\ref{obs:KP-to-height=1}) & para-NP-h (Prop.~\ref{pps:stars-obs}\ref{obs:KP-to-height=1})\\
		$\val_\lambda$ & para-NP-h (Thm.~\ref{thm:X-W1hard}) & para-NP-h (Thm.~\ref{thm:X-W1hard})\\
		$D+\wcode$ & open & FPT $\Oh(D \cdot 2^{\wcode} \cdot \numP)$ (Prop.~\ref{pps:height=1->mckp})\\
		$B+\var_w$ & XP $\Oh(B \cdot |X|^{2\cdot \var_w + 1})$ (Thm.~\ref{thm:B+varw}) & PFPT (Prop.~\ref{pps:height=1->mckp})\\
		$D+\var_w$ & para-NP-h (Prop.~\ref{pps:stars-obs}\ref{obs:D=height=1,varw=2}) & para-NP-h (Prop.~\ref{pps:stars-obs}\ref{obs:D=height=1,varw=2})\\
		\hline
	\end{tabular}
\end{table}

Further, we consider the \textit{maximum number of projects per taxon~$L:=\max_{i} \bet{P_i}$}. By definition,~$L=2$ in~\NAP[a_i]{c_i}{b_i}{2} and in~\GNAP we have~$L\le C+1$.
We denote the \textit{number of projects} by~$\numP=\sum_{i} |P_i|$. Clearly,~$\bet X\le \numP$, $L\le \numP$, and~$\numP \le \bet X \cdot L$.
By~$\var_c$, we denote the \textit{number of different costs}, that is,~$\var_c:=|\{ c_{i,j} : (c_{i,j},w_{i,j})\in P_i, P_i\in \mathcal{P} \}|$. We define the \textit{number of different \sprobs}~$\var_w$ analogously. This so-called \textit{number of numbers} parameterization was introduced by Fellows et al.~\cite{fellows12}; it is motivated by the idea that in many real-life instances this parameter may be small. In fact, for nonnegative numbers, the number of numbers is never larger than the maximum value which is used in pseudopolynomial-time algorithms.
Also, we consider the~\textit{maximum encoding length for \sprobs $\wcode=\max_{i,j}(\text{binary length of } w_{i,j})$} and the~\textit{maximum edge weight~$\val_\lambda = \max_{e\in E} \lambda(e)$}. Observe that because the maximal \sprob of a taxon could be smaller than 1, one cannot assume that~$\val_\lambda\le D$.

%\kommentar{\paragraph{Results}}{}
\begin{table}[t]
	\centering
%	\footnotesize
	\caption{Complexity results  for \NAP[0]{c_i}{1}{2}, the special case where the \sprobs are only 0 or 1, and \unitc the special case where each project has unit costs.
	The ``---''-sign indicates parameters that are (partially) constant in the specific problem definition and thus are not interesting.}
	\label{tab:results-special gnap}
	\begin{tabular}{lll}
		\hline
		Parameter & \NAP[0]{c_i}{1}{2} & \unitc \\
		\hline
		$|X|$ & FPT (Prop.~\ref{pps:BruteForce}\ref{obs:FPT-X-2-NAP}) & FPT (Prop.~\ref{pps:BruteForce}\ref{obs:FPT-X-2-NAP})\\
		$B$ & PFPT $\Oh(B^2 \cdot n)$ \cite{pardi} & XP (Prop.~\ref{pps:BruteForce}\ref{obs:XP-B})\\
		$C$ & PFPT $\Oh(C^2 \cdot n^3)$ (Cor.~\ref{cor:C-01-NAP}) & --- \\
		$D$ & PFPT $\Oh(D^2 \cdot n)$ (Prop.~\ref{pps:wcode=1-D}) & open \\
		\hline
		$\val_\lambda$ & PFPT $\Oh((\val_\lambda)^2 \cdot n^3)$ (Cor.~\ref{cor:wcode=1,val-lambda}) & open \\
		$\var_c$ & XP (Cor.~\ref{cor:wcode=1-varc}) & --- \\
		$\var_w$ & --- & XP (Cor.~\ref{cor:C=1-XP-varw})\\
		% $D+\wcode$ & --- & open \\
		% $B+\var_w$ & --- & XP (Cor.~\ref{cor:C=1-XP-varw})\\
		% $D+\var_w$ & --- & open \\
		\hline
	\end{tabular}
\end{table}

\subsection{Basic Observations}
We now observe some first complexity classifications that can be obtained from naive brute-force algorithms and from adaptions of a known NP-hardness reduction.
\begin{proposition}
	\label{pps:BruteForce}
	\begin{enumerate}[(a)]
		\item\label{pps:GNAP-X-XP}\GNAP can be solved in~$\Oh(L^{|X|} \cdot n^2)$~time.
		\item\label{obs:FPT-X-2-NAP}\NAP[a_i]{c_i}{b_i}{2} can be solved in~$\Oh(2^{\bet X} \cdot n^2)$~time.
		\item\label{obs:XP-B}\GNAP can be solved in~$||P||^B \cdot |\Instance|^{\Oh(1)}$~time.
	\end{enumerate}
\end{proposition}
\begin{proof}
	(\ref{pps:GNAP-X-XP})
	Consider all possibilities to select a project~$p_{i,j_i}\in P_i$ for all taxa~$x_i$. For each chosen set~$S:=\{p_{i,j_i} \mid i\in[|X|]\}$ compute~$\Costs(S)$ and~$PD_{\Tree}(S)$ and return yes if~$\Costs(S) \le B$ and~$PD_{\Tree}(S) \ge D$. If no considered set~$S$ gives a solution, then return no.
	
	For each taxon, the number of different choices is at most~$L$.  Moreover, we can compute~$\Costs(S)$ and $PD_{\Tree}(S) \ge D$ in~$\Oh(n^2)$~time.  Hence, the algorithm has running time~$\Oh(L^{|X|} \cdot n^2)$.
	
	(\ref{obs:FPT-X-2-NAP}) This is a direct consequence of the previous with~$L=2$.
	
	(\ref{obs:XP-B}) A~$\GNAP$ solution contains at most~$B$ projects with positive costs.
	Hence, a solution can be found by iterating over all size-$B$ project sets, checking whether any of them gives a solution.
\end{proof}

Some hardness results can be obtained directly from a known reduction from the NP-hard~\KP problem to~\NAP[0]{c_i}{1}{2}~\cite{pardi}. In~\KP one is given a set of items~$N$, a cost-function~$c: N \to \mathbb N$, a value-function~$d: N \to \mathbb N$, and two integers~$B$ and~$D$ and asks whether there is an item set~$N'$ such that~$c_\Sigma(N') \le B$ and~$d_\Sigma(N') \ge D$.

\begin{proposition}
	\label{pps:stars-obs}
	\begin{enumerate}[(a)]
		\item\label{obs:KP-to-height=1}\NAP[0]{c_i}{1}{2} is \NP-hard, even if the tree~$\Tree$ has height~1~\cite{pardi}.
		\item\label{obs:D=height=1,varw=2}\NAP[0]{c_i}{b}{2} is \NP-hard, even if~$D=1$,~$b\in (0,1]$ is a constant, and the tree~$\Tree$ has height~1.
		\item\label{cor:D=height=1,L=2,ultrametric}\NAP[0]{c_i}{b_i}{2} is \NP-hard, even if the tree~$\Tree$ is ultrametric with~$\height_{\Tree}=\val_\lambda=1$, and~$D=1$.
	\end{enumerate}
\end{proposition}
\begin{proof}
	(\ref{obs:KP-to-height=1})
	Let~$\Instance=(N,c,d,B,D)$ be an instance of~\KP.
	Define~$\Tree:=(V,E,\lambda)$ to be an~$N$-tree with~$V:=\{w\} \cup N$ and~$E:=\{ (w,x_i) \mid x_i\in N \}$ and~$\lambda(w,x_i) := d(x_i)$.
	For each leaf~$x_i$, define a project list~$P_i$ that contains two projects~$(0,0)$ and~$(c(x_i),1)$.
	
	Then,~$\Instance':=(\Tree,\mathcal{P},B':=B,D':=D)$ is a yes-instance of~\NAP[0]{c_i}{1}{2} if and only if~$\Instance$ is a yes-instance of~\KP. This can be seen by observing that a set~$S$ is a solution for~$\Instance$ if and only if the project set $S'$ obtained by adding $p_{i,1}
	=(0,0)$  if $x_i\notin S$ and $p_{i,2}=(c(x_i),1)$ if $x_i\in S$ is a solution for~$\Instance'$.
	Indeed this follows from
	$$
	d_\Sigma(S)= \sum_{x_i \in S} d(x_i) = \sum_{x_i \in S} \lambda(w,x_i)=PD_\Tree(S')
	$$
	and
	$$
	c_\Sigma(S) = \sum_{x_i \in S} c(x_i)= \sum_{x_i \in S}  c_{i,2} =\Costs(S').
	$$
	
	(\ref{obs:D=height=1,varw=2})
	In this reduction, one could also set~$D':=1$ and replace the \sprob of every project with positive cost to~$\nicefrac 1D$.
	Thus,~\SGNAP is \NP-hard even if~$D=1$ and~$\var_w=2$.
	
	(\ref{cor:D=height=1,L=2,ultrametric})
	We may assume in the previous reduction that there is no project~$p = (c,1/D)$ of taxon~$x$ with~$\lambda(\rho,x) > 1/D$ and~$c<B$, as we have a trivial yes-instance, otherwise.
	Henceforth, we remove all taxa with project cost~$c>B$ and for all remaining projects, we divide the \sprob by~$\lambda(\rho,x)$, directly yielding the desired result.
\end{proof}

\section{Multiple-Choice Knapsack}
\label{sec:MCKP}
In this section, we consider \MCKPLong (\MCKP), a classic variant of \KP in which the item set is divided into classes and from every class, exactly one item can be chosen. \MCKP has been studied numerous times over the years~\cite{nauss78,pisinger,bansal,kellerer}. \MCKP is also special case of the \textsc{Max d-Dimensional Knapsack} problem, for which~Gurski, Rehs, and Rethmann~\cite{rehs} provided a parameterized complexity analysis. 

To define the problem, let~$c$ and~$d$ denote the cost and value functions, respectively. Accordingly, for an item~$a_i$, we call~$c(a_i)$ the {\it cost of~$a_i$} and~$d(a_i)$ the {\it value of~$a_i$}. Recall that using our function notation for sets, this means that~$c_\Sigma(A)$ and~$d_\Sigma(A)$ denote the sum of the costs and values of the elements of~$A$ and~$c(A)$ and~$d(A)$ denote the set of costs and values assumed by some element of~$A$.  
\Problem{\MCKPLong (MCKP)}
{A set of items~$N=\{a_1,\dots,a_n\}$,
	a partition~$\{N_1,\dots,N_m\}$ of~$N$,
	two functions~$c,d: N\to \mbb N$,
	and two integers~$B$ and~$D$}
{Is there a set~$S\subseteq N$ such that
	$|S\cap N_i|=1$ for each~$i\in [m]$,
	$c_\Sigma(S)\le B$, and~$d_\Sigma(S)\ge D$}
A set~$S$ that fulfills these criteria is called a~\textit{solution} for the instance~$\mathcal I$.

For \MCKP, we consider parameters that are closely related to the parameters described for \GNAP: The input directly gives the \textit{number of classes~$m$}, the \textit{budget~$B$}, and the desired~\textit{value~$D$}.
Closely related to~$B$ is the \textit{maximum cost for an item~$C=\max_{a_j \in N} c(a_j)$}. Since projects whose cost exceeds the budget can be removed from the input, we may assume~$C\le B$. Further, we assume that~$B\le C\cdot m$, as otherwise, one can return yes if the total value of the most valuable items per class exceeds~$D$, and no otherwise.
We also consider~$\var_c := |c(N)|$, the \textit{number of different costs}, and~$\var_d:=|d(N)|$, the \textit{number of different values}.
The size of the biggest class is denoted by~$L$. If one class~$N_i$ contains two items~$a_p$ and~$a_q$ with the same cost and~$d(a_p)\le d(a_q)$, the item~$a_p$ can be removed from the instance. Thus, we may assume that no class contains two items with the same cost, and so~$L\le \var_c$. Analogously, we may assume that no class contains two same-valued items and so~$L\le\var_w$.
%\todos{We may say something like that the most parameters for \MCKP can be seen analogously like in \GNAP to save space.}
Table~\ref{tab:results-mckp} lists the old and new complexity results for \MCKPLong.

\begin{table}[t]
	\centering
	\caption{Complexity results for \MCKPLong.}
%		\footnotesize
		\label{tab:results-mckp}
	\begin{tabular}{ll}
		\hline
		Parameter & \MCKP\\
		\hline
		$m$ & \W1-hard, XP (Thm.~\ref{thm:MCKP-m-W1})\\
		$B$ & PFPT $\Oh(B \cdot |N|)$~\cite{pisinger}\\
		$C$ & PFPT $\Oh(C \cdot |N|\cdot m)$ (Prop.~\ref{obs:C-MCKP})\\
		$D$ & PFPT $\Oh(D \cdot |N|)~$\cite{bansal}\\
		$L$ & para-NP-h~\cite{kellerer}\\
		$\var_c$ & XP~$\Oh(m^{\var_c-1}\cdot |N|)$ (Prop.~\ref{pps:MCKP-XP-varc})\\
		$\var_d$ & XP~$\Oh(m^{\var_d-1}\cdot |N|)$ (Prop.~\ref{pps:MCKP-XP-vard})\\
		$\var_c+\var_d$ & \FPT (Thm.~\ref{thm:MCKP-ILPF})\\
		\hline
	\end{tabular}
\end{table}\subsection{Algorithms for Multiple-Choice Knapsack}
\label{subsec:algos-MCKP}
First, we provide some algorithms for~\MCKP. 
It is known that~\MCKP can be solved in~$\Oh(B \cdot |N|)$ time~\cite{pisinger}, or in~$\Oh(D \cdot |N|)$ time~\cite{bansal}.
As we may assume that~$C\cdot m\ge B$, we may also observe the following.
\begin{proposition}
	\label{obs:C-MCKP}
	\MCKP can be solved in~$\Oh(C \cdot \bet N \cdot m)$ time.
\end{proposition}

\KP is FPT with respect to the number of different costs~$\var_c$~\cite{etscheid}.
This result is shown by a reduction to \ILPF with~$f(\var_c)$ variables.
This approach cannot be adopted easily, as it has to be checked whether a solution contains exactly one item per class.
In Propositions~\ref{pps:MCKP-XP-varc} and~\ref{pps:MCKP-XP-vard} we show that \MCKP is XP with respect to the number of different costs and different values, respectively.
Then, in Theorem~\ref{thm:MCKP-ILPF}, we show that \MCKP is FPT with respect to the parameter~$\var_c+\var_d$.

In the following, let~$\mcal I=(N,\{N_1,\dots,N_m\},c,d,B,D)$ be an instance of~\MCKP. We let~$c_1,\dots,c_{\var_c}$ and~$d_1,\dots,d_{\var_d}$ denote the set of different costs and the set of the different values in~\Instance where~$c_i<c_{i+1}$ for each~$i\in[\var_c-1]$ and~$d_j<d_{j+1}$ for each~$j\in [\var_d-1]$. In other words,~$c_i$ is the~$i$th cheapest cost in~$c(N)$ and~$d_j$ is the~$j$th smallest value in~$d(N)$.
Since  each class~$N_i$ contains at most one item with cost~$c_p$ and at most one item with value~$d_q$, we have~$|N_i|\le \var_c$ and~$|N_i|\le \var_d$ for every~$i\in[m]$.

\begin{proposition}
	\label{pps:MCKP-XP-varc}
	\MCKP can be solved in~$\Oh(m^{\var_c-1}\cdot |N|)$ time, where~$\var_c$ is the number of different costs.
\end{proposition}
\begin{proof}
	We describe a dynamic programming algorithm with a table that has a dimension for all the~$\var_c$ different costs, except for~$c_{\var_c}$.
	An entry~$F[i,p_1,\dots,p_{\var_c-1}]$ of the table stores the maximal total value of a set~$S$ that contains exactly one item of each set of~$N_1,\dots,N_i$ and contains exactly~$p_j$ items with cost~$c_j$ for each~$j\in [\var_c-1]$. Subsequently,~$S$ contains exactly~$p_{\var_c}^{(i)} := i-\sum_{j=1}^{\var_c-1} p_j$ items with cost~$c_{\var_c}$.
	If~$\sum_{j=1}^{\var_c-1} p_j \ge i$, then store~$F[i,p_1,\dots,p_{\var_c-1}] = -\infty$.
	In the rest of the proof, by~$\myvec{p}$ we denote~$(p_1,\dots,p_{\var_c-1})$ and by~$\myvec{p}_{(j) +z}$ we denote~$(p_1,\dots,p_j+z,\dots,p_{\var_c-1})$ for an integer~$z$.
	Let~${\myvec{0}}$ denote the~$(\var_c-1)$-dimensional zero.
	
	For the initialization of the values for~$F[1,\myvec{p}]$, only size-1 subsets of~$N_1$ are considered.
	Thus, for every~$a\in N_1$ with cost~$c(a) = c_j < c_{\var_c}$, store~$F[1,{\myvec{0}}_{(j) + 1}]=d(a)$.
	If~$N_1$ contains an item~$a$ with cost~$c(a)=c_{\var_c}$, then store~$F[1,{\myvec{0}}]=d(a)$.
	For all other~$\myvec{p}$, store~$F[1,\myvec{p}]=-\infty$.
	Once the entries for~$i$ classes have been computed, one may compute the entries for~$i+1$ classes using the recurrence
	\begin{eqnarray}
		\label{eqn:MCKP-varc}
		F[i+1,\myvec{p}] &=&
		\max_{a \in N_{i+1}}
		\left\{
		\begin{array}{ll}
			F[i,\myvec{p}_{(j) -1}] + d(a)
			& \text{if } c(a) = c_j < c_{\var_c} \text{ and } p_j \ge 1\\
			F[i,\myvec{p}] + d(a)
			& \text{if } c(a) = c_{\var_c}
		\end{array}
		\right.
		.
	\end{eqnarray}
	
	Return yes if~$F[m,\myvec{p}]\ge D$ for~$p_{\var_c}^{(m)} \cdot c_{\var_c} + \sum_{i=1}^{\var_c-1} p_i \cdot c_i \le B$ and return no, otherwise.
	
	\pfpara{Correctness}
	For given integers~$i\in [m]$ and~$\myvec{p} \in [i]_0^{\var_c-1}$, we define~$\mathcal S^{(i)}_{\myvec{p}}$ to be the family of~$i$-sized sets~$S\subseteq N$ that contain exactly one item of each of~$N_1,\dots,N_i$ and where~$p_\ell$ is the number of items in~$S$ with cost~$c_\ell$ for each~$\ell \in [\var_c-1]$.
	
	For fixed~$\myvec{p}\in [i]_0^{\var_c-1}$, we prove that~$F[i,\myvec{p}]$ stores the maximal value of a set~$S\in \mathcal S^{(i)}_{\myvec{p}}$ by an induction. This implies that the algorithm is correct.
	For~$i=1$, the claim is easy to verify from the initialization.
	Now, assume the claim is correct for a fixed~$i\in [m-1]$. We first prove that if~$F[i+1,\myvec{p}]=q$, then there exists a set~$S\in \mathcal S^{(i+1)}_{\myvec{p}}$ with~$d_\Sigma(S)=q$.
	Afterward, we prove that~$F[i+1,\myvec{p}]\ge d_\Sigma(S)$ for every set~$S\in \mathcal S^{(i+1)}_{\myvec{p}}$.
	
	($\Rightarrow$)
	Let~$F[i+1,\myvec{p}]=q$.
	Let~$a\in N_{i+1}$ with~$c(a) = c_j$ be an item of~$N_{i+1}$ that maximizes the right side of Equation~(\ref{eqn:MCKP-varc}) for~$F[i+1,\myvec{p}]$.
	In the first case, assume~$c_j < c_{\var_c}$, and thus~$F[i+1,\myvec{p}]=q=F[i,\myvec{p}_{(j) -1}] + d(a)$.
	By the induction hypothesis, there is a set~$S\in S^{(i)}_{\myvec{p}_{(j) -1}}$ such that~$F[i+1,\myvec{p}_{(j) -1}]=d_\Sigma(S)=q-d(a)$.
	Observe that~$S':=S\cup\{a\}\in S^{(i+1)}_{\myvec{p}}$. The value of~$S'$ is~$d_\Sigma(S')=d_\Sigma(S)+d(a)=q$.
	The other case with~$c(a) = c_{\var_c}$ is shown analogously.
	
	($\Leftarrow$)
	Let~$S\in \mathcal S^{(i+1)}_{\myvec{p}}$ and let~$a\in S\cap N_{i+1}$.
	In the first case, assume~$c(a) = c_j < c_{\var_c}$.
	Observe that~$S' := S\setminus\{a\}\in S^{(i)}_{\myvec{p}_{(j) -1}}$.
	Consequently,
	\begin{eqnarray}
		\label{eqn:MCKP-varc-IH1}
		F[i+1,\myvec{p}]
		&\ge& F[i,\myvec{p}_{(j) -1}] + d(a)\\
		\label{eqn:MCKP-varc-IH2}
		&=& \max\{ d_\Sigma(S) \mid S \in S^{(i)}_{\myvec{p}_{(j) -1}}\} + d(a)\\
		\nonumber
		&\ge& d_\Sigma(S') + d(a) = d_\Sigma(S).
	\end{eqnarray}
	Herein, Inequality~(\ref{eqn:MCKP-varc-IH1}) is the definition of the recurrence in Equation~(\ref{eqn:MCKP-varc}), and Equation~(\ref{eqn:MCKP-varc-IH2}) follows by the induction hypothesis.
	The other case with~$c(a) = c_{\var_c}$ is shown analogously.
	
	\pfpara{Running time}
	First, we show how many options of~$\myvec{p}$s there are and then how many equations have to be computed for one of these options.
	
	For~$p_1,\dots,p_{\var_c-1}$ with~$\sum_{j=1}^{\var_c-1} p_j \ge m$ and each~$i\in [m]$, the entry~$F[i,\myvec{p}]$ stores~$-\infty$. Consequently, we consider a vector~$\myvec{p}$ with~$p_j=m$ only if~$p_\ell=0$ for each~$\ell \ne j$.
	Thus, we are only interested in~$\myvec{p}\in [m-1]_0^{\var_c-1}$ or~$\myvec{p}={\myvec{0}}_{(j) +m}$ for each~$j\in[\var_c-1]$.
	So, there are~$m^{\var_c-1} + m \in \Oh(m^{\var_c-1})$ options of~$\myvec{p}$.
	
	For a fixed~$\myvec{p}$, each item~$a\in N_i$ is considered exactly once in the computation of~$F[i,\myvec{p}]$. Thus, overall~$\Oh(m^{\var_c-1} \cdot |N|)$ time is needed to compute the table~$F$.
	Additionally, we need~$\Oh(\var_c \cdot |N|)$ time to check whether~$F[m,p_1,\dots,p_{\var_c-1}]\ge D$ and~$p_{\var_c}^{(m)} \cdot c_{\var_c} + \sum_{i=1}^{\var_c-1} p_i \cdot c_i \le B$ for any~$\myvec{p}$. As we may assume~$m^{\var_c-1} > \var_c$, the running time of the entire algorithm is~$\Oh(m^{\var_c-1} \cdot |N|)$.
\end{proof}

We now adapt of the dynamic programming algorithm described in the proof of Proposition~\ref{pps:MCKP-XP-varc}.
In this adaption, instead of storing the maximum value of a set of items with a given set of costs, we store the minimum cost a set of items with a given set of values. This shows that the problem is in XP with respect to the number of different values of the input items.

\begin{proposition}
	\label{pps:MCKP-XP-vard}
	\MCKP{} can be solved in~$\Oh(m^{\var_d-1}\cdot |N|)$, where~$\var_d$ is the number of different values.
\end{proposition}
\begin{proof}
	We describe a dynamic programming algorithm with a table that has a dimension for all the~$\var_d$ different values except for~$d_{\var_d}$.
	An entry~$F[i,p_1,\dots,p_{\var_d-1}]$ of the table stores the minimal cost of a set~$S$ that contains exactly one item of each set of~$N_1,\dots,N_i$ and contains exactly~$p_j$ items with value~$d_j$ for each~$j\in [\var_d-1]$. Subsequently,~$S$ contains exactly~$p_{\var_d}^{(i)} := i-\sum_{j=1}^{\var_d-1} p_j$ items with value~$d_{\var_d}$.
	If~$\sum_{j=1}^{\var_d-1} p_j \ge i$, then store~$F[i,p_1,\dots,p_{\var_d-1}] = B+1$.
	In the rest of the proof, by~$\myvec{p}$ we denote~$(p_1,\dots,p_{\var_d-1})$ and by~$\myvec{p}_{(j) +z}$ we denote~$(p_1,\dots,p_j+z,\dots,p_{\var_d-1})$ for an integer~$z$.
	Let~${\myvec{0}}$ denote the~$(\var_d-1)$-dimensional zero.
	
	For the initialization of the values for~$F[1,\myvec{p}]$, only size-1 subsets of~$N_1$ are considered.
	Thus, for every~$a\in N_1$ with value~$d(a) = d_j < d_{\var_d}$, store~$F[1,{\myvec{0}}_{(j) + 1}]=c(a)$.
	If~$N_1$ contains an item~$a$ with value~$d(a)=d_{\var_d}$, then store~$F[1,{\myvec{0}}]=c(a)$.
	For all other~$\myvec{p}$ store~$F[1,\myvec{p}]=B+1$.
	Once the entries for~$i$ classes have been computed, one may compute the entries for~$i+1$ classes using the recurrence
	\begin{eqnarray}
		\label{eqn:MCKP-vard}
		F[i+1,\myvec{p}] &=&
		\max_{a \in N_{i+1}}
		\left\{
		\begin{array}{ll}
			F[i,\myvec{p}_{(j) -1}] + c(a)
			& \text{if } d(a) = d_j < d_{\var_d} \text{ and } p_j \ge 1\\
			F[i,\myvec{p}] + c(a)
			& \text{if } d(a) = d_{\var_d}
		\end{array}
		\right.
		.
	\end{eqnarray}
	
	Return yes if~$F[m,\myvec{p}]\le B$ for~$p_{\var_d}^{(m)} \cdot d_{\var_d} + \sum_{i=1}^{\var_d-1} p_i \cdot d_i \ge D$ and return no, otherwise.

		The correctness and running time proof are completely analogous to the proof of Proposition~\ref{pps:MCKP-XP-varc}.
\end{proof}

By Propositions~\ref{pps:MCKP-XP-varc} and~\ref{pps:MCKP-XP-vard}, \MCKP is~\XP with respect to~$\var_c$ and~$\var_d$, respectively. In the following, we reduce instances of \MCKP to instances of~\ILPF with at most~$2^{\var_c+\var_d} \cdot \var_c$ variables. Since \ILPF instances with~$n$ variables and input length~$s$ can be solved using~$s \cdot n^{2.5n+o(n)}$ arithmetic operations~\cite{frank,lenstra}, this shows that~\MCKP is fixed-parameter tractable with respect to the combined parameter~$\var_c+\var_d$.
\begin{theorem}
	\label{thm:MCKP-ILPF}
	For each instance of~\MCKP, an equivalent instance of~\ILPF\linebreak with~$\Oh(2^{\var_c+\var_d} \cdot \var_c)$ variables can be computed in polynomial time. Thus,~\MCKP is \FPT with respect to~$\var_c+\var_d$.
\end{theorem}
%It is known that~\KP is \FPT with respect to~$\var_c$~\cite{etscheid}. We were not able to use the technique for this result to prove that also \MCKP is \FPT with respect to $\var_c$.
\begin{proof}
	\pfpara{Reduction}
	We may assume that a class~$N_i$ does not contain two items of the same cost or the same value and that for items~$a,b\in N_i$ with~$c(a)<c(b)$ we have~$d(a)<d(b)$.
	Thus, each class~$N_i$ is fully described by the set of costs~$c(N_i)$ of the items in~$N_i$ and the set of values~$d(N_i)$ of the items in~$N_i$.
	In the following, we call~$T=(C,Q)$ a \textit{type}, for sets~$C\subseteq c(N), Q\subseteq d(N)$ with~$|C| = |Q|$. Let~$\Tree$ be the family of all types.
	We say that the \textit{class~$N_i$ is of type~$T=(c(N_i),d(N_i))$.}
	For each~$T\in\Tree$, let~$m_T$ be the number of classes of type~$T$.
	Clearly,~$\sum_{T\in\Tree} m_T = m$.
	
	Now, for each class~$N_j$ of type~$T=(C,Q)$, the cost~$c(a)$ of an item~$a\in N_j$ directly determines its value~$d(a)$. More precisely, if~$c(a)$ is the~$\ell$th cheapest cost in~$C$, then the value of~$a$ is the~$\ell$th smallest value in~$Q$. Accordingly, if~$c_i$ is the~$\ell$th smallest cost in~$C$, we let~$d_{T,i}$ denote the~$\ell$th smallest value in~$Q$. For all~$i\in [\var_c]$ where~$c_i \not\in C$, we define~$d_{T,i}:=-\sum_{i=1}^m \max d(N_i)$. 
	
	We now define the following instance of~\ILPF that is equivalent to~\Instance.
	The variable~$x_{T,i}$ expresses the number of items with cost~$c_i$ that are chosen in a class of type~$T$.
	\begin{align}
		\label{eqn:MCKP-ILPF-B}
		\sum_{T \in \mathcal{T}_C} \sum_{i=1}^{\var_c} x_{T,i} \cdot c_i \le & \; B\\
		\label{eqn:MCKP-ILPF-D}
		\sum_{T \in \mathcal{T}_C} \sum_{i=1}^{\var_c} x_{T,i} \cdot d_{T,i} \ge & \; D\\
		\label{eqn:MCKP-ILPF-Nr}
		\sum_{i=1}^{\var_c} x_{T,i} = & \; m_{T} & \quad \forall T\in \mathcal{T}\\
		\label{eqn:MCKP-ILPF-initial}
		x_{T,i} \ge & \; 0 & \quad \forall T\in \mathcal{T}, i\in [\var_c]
	\end{align}
	
	\pfpara{Correctness}
	Observe that if~$c_i\not\in C$, then Inequality~(\ref{eqn:MCKP-ILPF-D}) would not be fulfilled if~$x_{T,i}>0$  because we defined~$d_{T,i}$ to be~$-\sum_{i=1}^m \max d(N_i)$. Consequently,~$x_{T,i}=0$ if~$c_i\not\in C$ for each type~$T=(C,Q)\in\Tree$ and~$i\in[\var_c]$.
	Inequality~(\ref{eqn:MCKP-ILPF-B}) ensures that the total cost is at most~$B$.
	Inequality~(\ref{eqn:MCKP-ILPF-D}) ensures that the total value is at least~$D$.
	Equation~(\ref{eqn:MCKP-ILPF-Nr}) ensures that exactly~$m_T$ elements are picked from the classes of type~$T$ for each~$T\in \Tree$.
	Finally, observe that the constructed instance has~$\Oh(2^{\var_c+\var_d}\cdot \var_c)$ variables because~$\Tree\subseteq 2^{c(N)} \times 2^{d(N)}$ which gives~$\Oh(2^{\var_c+\var_d}\cdot \var_c)$ different options for the variables~$x_{T,i}$.
\end{proof}
%Observe that with the same technique, an instance of~\ILPF with~$2^{\var_c+\var_d}\cdot \var_d$ variables can be described.

\subsection{Hardness with respect to the Number of Classes}
Second, we contrast the algorithms by the following hardness results.
There is a reduction from \KP{} to~\MCKP in which each item in the instance of~\KP{} is added to a unique class with a further, new item that has no costs and no value~\cite{kellerer}.
This reduction constructs an instance with many classes.
In the following, we prove that~\MCKP is \W{1}-hard with respect to the number of classes~$m$, even if~$B=D$ and~$c(a)=d(a)$ for each~$a\in N$. This special case of~\MCKP is called~\textsc{Multiple-Choice \SubSum}~\cite{kellerer}.

To show the hardness, we reduce from \SubSum with multi-set input parameterized by~$k$.
\Problem{Subset Sum}
{A multi-set of integers~$Z=\{z_1,\dots,z_n\}$ and integers~$Q$ and~$k$}{Is there a multi-set~$S\subseteq Z$ such that~$\sum_{s\in S}=Q$ and~$|S|=k$}

In parameterized complexity, the problem is often defined for set inputs. The original W[1]-hardness proof for \SubSum relies on a reduction from the \textsc{Perfect Code} problem~\cite[Lemma~4.4]{DF95}. It is easy to observe that this reduction also works if every constructed integer is added~$k$ times to~$Z$. We will not repeat the details of the reduction but give a brief intuition. In the reduction, the target number~$Q$ has a value of one at each digit. Now, adding~$k$ copies of a number to~$Z$ in the construction maintains correctness because including any number twice in the solution produces carries in the summation which destroys the property that every digit has a value of~1.

\begin{proposition}[\cite{DF95}]
	\label{prop:MultiSubSum}
	\SubSum{} is \W{1}-hard with respect to~$k$, even when every integer in~$Z$ has multiplicity at least~$k$ in~$Z$. 
\end{proposition}

\begin{theorem}
	\label{thm:MCKP-m-W1}
	\MCKP is \XP and \W1-hard with respect to the number of classes~$m$.
\end{theorem}
\begin{proof}
	The \XP-algorithm follows from the fact that we may simply try all~$|N|^m$ possibilities to select a set~$S$ of items and compute the value and cost for each of them.
	
	\pfpara{Reduction}
	We reduce from~\SubSum with multi-set inputs~$Z$ where every integer~$z$ occurs~$k$ times.
	Let~$\Instance=(Z,Q,k)$ be such an instance of~\SubSum
	and assume, without loss of generality, that~$Z=\{z_{i,j}\mid i\in n , j\in [k]\}$ such that~$z_{i,1}=z_{i,2}=\ldots =z_{i,k}$ for all~$i\in [n]$. 
	
	We define~$N_j:=\{a_{i,j}\mid i\in [n]\}$ for each~$j\in [k]$ and set~$c(a_{i,j})=d(a_{i,j})=z_{i,j}$ for each element $a_{i,j}$. Then~$\Instance'$ is~$(N,\{N_1,\ldots, N_k\},Q,Q)$ where~$N$ is the union of the sets~$N_j$,~$j\in [k]$. Observe that the reduction can be performed in polynomial time and that~$m=k$, so the reduction is indeed a parameterized reduction from \SubSum with parameter~$k$ to \MCKP with parameter~$m$. 
	
	\pfpara{Correctness}
	We show that~\Instance is a yes-instance if and only if~\Instance' is a yes-instance.
	
	($\Rightarrow$) Let~$S=\{z_{i_1},\ldots,z_{i_k}\}$ be a solution for~$\Instance$. By construction,~each set~$N_j$ contains one element~$a_{\ell_j}$ such that~$c(a_{\ell_j})=d(a_{\ell_j})=z_{i_j}$. Then, the set~$S':=\{a_{\ell_1}, a_{\ell_2},\ldots, a_{\ell_k} \}$ is a solution for~$\Instance'$ since $$\sum_{j\in [k]} c(a_{\ell_j}) = \sum_{j\in [k]} d(a_{\ell_j}) = \sum_{j\in [k]} z_{i_j}=Q. $$
	
	($\Leftarrow$)
	Conversely, let~$S':=\{a_{\ell_1}, a_{\ell_2},\ldots, a_{\ell_k} \}$ be a solution for~$\Instance'$. By construction, for each~$j\in [k]$, there is a number~$z_{i_j}\in Z$ such that~$c(a_{\ell_j})=z_{i_j}$.  Then,~$S:=\{z_{i_j}\mid j\in [k] \}$ is a solution for~$\Instance$ since
	$$Q=\sum_{j\in [k]} c(a_{\ell_j}) = \sum_{j\in [k]}  z_{i_j}. $$
\end{proof}

\section{The Generalized Noah's Ark Problem}
\label{sec:GNAP}
We now come back to the maximization of phylogenetic diversity, starting with the most general problem variant, \GNAPLong (\GNAP).
Subsequently, we turn to \SGNAP, making use of the tractability and hardness results for \MCKP established above.

\subsection{Algorithms for the Generalized Noah's Ark Problem}

In Theorem~\ref{thm:varc+varw}, we now show that \GNAP can be solved in polynomial time when the number of different project costs and the number of different survival probabilities is constant.
In the following, let~$\Instance=(\Tree,\mathcal P,B,D)$ be an instance of~\GNAP, and let~$\mathcal C:=\{c_1,\dots,c_{\var_c}\}$ and~$\mathcal W:=\{w_1,\dots,w_{\var_w}\}$ denote the sets of different costs and different \sprobs in~\Instance, respectively.
Without loss of generality, assume~$c_i<c_{i+1}$ for each~$i\in[\var_c-1]$ and likewise assume~$w_j<w_{j+1}$ for each~$j\in [\var_w-1]$. In other words,~$c_i$ is the~$i$th cheapest cost in~$\mathcal C$, and~$w_j$ is the~$j$th smallest \sprob in~$\mathcal W$.
Recall, that we assume that there is at most one item with cost~$c_p$ and at most one item with \sprob~$w_q$ in every project list~$P_i$, for each~$p\in[\var_c]$ and~$q\in[\var_w]$.
For the rest of the section, by~$\myvec{a}$ and~$\myvec{b}$ we denote~$(a_1,\dots,a_{\var_c-1})$ and~$(b_1,\dots,b_{\var_w-1})$, respectively.
Further, we let~$\myvec{p}_{(j) +z}$ denote the vector~$\myvec{p}$ in which, at position~$i$, the value~$z$ is added, and we let~${\myvec{0}}$ denote the~$(\var_c-1)$-dimensional zero.

\begin{theorem}
	\label{thm:varc+varw}
	\GNAP can be solved in~$\Oh\left(|X|^{2(\var_c+\var_w-1)}\cdot (\var_c+\var_w)\right)$ time.
\end{theorem}
\begin{proof}
	We describe a dynamic programming algorithm with two tables,~$F$~and~$G$, that have a dimension for all the~$\var_c$ different costs, except for~$c_{\var_c}$ and all the~$\var_w$ different \sprobs, except for~$\var_w-1$.
	Recall that for a vertex $v$ with $t$ children, $T_v$ is the subtree rooted at $v$ and
	the offspring $\off(v)$ of $v$ are the leaves in $T_v$.
	For a vertex~$v$ with children $w_1,\dots,w_t$ where~$w_i$ denotes the~$i$th child, the \textit{$i$-partial subtree $T_{v,i}$ rooted at~$v$} for $i\in [t]$ is the  connected component containing~$v$ in~$T_v - \{(v,w_{i+1}), \ldots, (v,w_{t})\}$.
	For a vertex~$v\in V$ and given vectors~\myvec{a} and~\myvec{b}, we define~$\mathcal S^{(v)}_{\myvec{a},\myvec{b}}$ to be the family of sets of projects~$S$ such that
	\begin{enumerate}
		\item $S$ contains exactly one project of~$P_i$ for each~$x_i\in \off(v)$,
		\item $S$ contains exactly~$a_k$ projects of cost~$c_k$ for each~$k \in [\var_c-1]$, and
		\item $S$ contains exactly~$b_\ell$ projects of \sprob~$w_\ell$ for each~$\ell \in [\var_w-1]$.
	\end{enumerate}
	For a vertex~$v\in V$ with children~$u_1,\dots,u_t$, given vectors~\myvec{a}, and~\myvec{b} and a given integer~$i\in[t]$ we define~$\mathcal S^{(v,i)}_{\myvec{a},\myvec{b}}$ analogously, just that exactly one project of $P_i$ is chosen for each~$x_i\in \off(u_1)\cup\dots\cup\off(u_i)$.
	
	It follows that we can compute how many projects with cost~$c_{\var_c}$ and \sprob~$w_{\var_w}$ a set~$S\in \mathcal S^{(v)}_{\myvec{a},\myvec{b}}$ contains. There are~$a_{\var_c}^{(v)} := |\off(v)|-\sum_{j=1}^{\var_c-1} a_j$ projects with cost~$c_{\var_c}$ and~$b_{\var_w}^{(v)} := |\off(v)|-\sum_{j=1}^{\var_w-1} b_j$ projects with \sprob~$w_{\var_w}$.
	The entries~$F[v,\myvec{a},\myvec{b}]$ and~$G[v,i,\myvec{a},\myvec{b}]$ store the maximum expected phylogenetic diversity of the tree $T_v$ for~$S\in \mathcal S^{(v)}_{\myvec{a},\myvec{b}}$ and $T_{v,i}$ for~$S\in\mathcal S^{(v,i)}_{\myvec{a},\myvec{b}}$, respectively.
	We further define the total \sprob to be~$w(b_{\var_w},\myvec{b}) := 1-
	(1-w_{\var_w})^{b_{\var_w}}
	\cdot
	\prod_{i=1}^{\var_w-1}
	(1-w_i)^{b_i}$,
	when~$b_{\var_w}$ and~$\myvec{b}$ describe the number of chosen single~\sprobs.
	
	Fix a taxon~$x_i$ with project list~$P_i$.
	As we want to select exactly one project of $P_i$, the project is clearly defined by $\vec a$ and $\vec b$.
	So, we store~\mbox{$F[x_i,\myvec{a},\myvec{b}] = 0$}, if~$P_i$ contains a project~$p = (c_k,w_\ell)$ such that
	\begin{enumerate}
		\item ($k<\var_c$ and~$\myvec{a}={\myvec{0}}_{(k)+1}$ or~$k=\var_c$ and~$\myvec{a}={\myvec{0}}$), and
		\item ($\ell<\var_w$ and~$\myvec{b}={\myvec{0}}_{(\ell)+1}$ or~$\ell=\var_w$ and~$\myvec{b}={\myvec{0}}$).
	\end{enumerate}
	Otherwise, store~$F[x_i,\myvec{a},\myvec{b}] = -\infty$.
	
	Let~$v$ be an internal vertex with children~$u_1,\dots,u_t$, we define
	\begin{eqnarray}
		\label{eqn:varc+varw-G-1}
		G[v,1,\myvec{a},\myvec{b}]
		&=&
		F[u_1,\myvec{a},\myvec{b}]
		+ \lambda(v,u_1) \cdot
		w\left(b_{\var_w}^{(u_{1})},\myvec{b}\right)
	\end{eqnarray}
	and to compute further values of~$G$, we can use the recurrence
	\begin{eqnarray}
		\label{eqn:varc+varw-G-i+1}
		G[v,i+1,\myvec{a},\myvec{b}]
		&=&
		\max_{
			\begin{array}{l}
				{\myvec{0}}\le \myvec{a'}\le \myvec{a}\\
				{\myvec{0}}\le \myvec{b'}\le \myvec{b}
		\end{array}}
		\left\{
		\begin{array}{l}
			G[v,i,\myvec{a}-\myvec{a'},\myvec{b}-\myvec{b'}]
			+ F[u_{i+1},\myvec{a'},\myvec{b'}]\\
			+ \lambda(v,u_{i+1}) \cdot
			w\left(b_{\var_w}^{(u_{i+1})},\myvec{b'}\right)
		\end{array}
		\right.
		.
	\end{eqnarray}
	Herein, we write~$\myvec{p}\le \myvec{q}$ if~$\myvec{p}$ and~$\myvec{q}$ have the same dimension~$d$ and~$p_i\le q_i$ for every~$i\in[d]$.
	And finally, we define~$F[v,\myvec{a},\myvec{b}] = G[v,t,\myvec{a},\myvec{b}]$.
	
	Return yes if there are~$\myvec{a}$ and~$\myvec{b}$ such that~$\sum_{i=1}^{\var_c-1} a_i \le |X|$, and~$\sum_{i=1}^{\var_w-1} b_i \le |X|$, and~$a_{\var_c}^{(r)} \cdot c_{\var_c} + \sum_{i=1}^{\var_c-1} a_i \cdot c_i \le B$, and~$F[r,\myvec{a},\myvec{b}]\ge D$ where~$r$ is the root of~$\Tree$. Otherwise, return no.
	
	\pfpara{Correctness}
	For any~$v$,~\myvec{a},~\myvec{b} and~$i$, we prove that~$F[v,\myvec{a},\myvec{b}]$ and~$G[v,i,\myvec{a},\myvec{b}]$ store~$\max PD_{\Tree_v}(S)$ for~$S\in \mathcal S^{(v)}_{\myvec{a},\myvec{b}}$ and~$S\in \mathcal S^{(v,i)}_{\myvec{a},\myvec{b}}$, respectively. This implies that the algorithm is correct.
	For a taxon~$x_i$, the tree $T_{x_i}$ does not contain edges and so there is no diversity. We can only check if \myvec a and \myvec b correspond to a feasible project. So, the table~$F$ stores the correct value by the initialization.
	For an internal vertex~$v$, the children~$u_1,\dots,u_t$ of~$v$ and~$i\in [t-1]$,
	the entry~$G[v,1,\myvec{a},\myvec{b}]$ stores the correct value by the observations that~$PD_{\Tree_{v,1}}(S) = PD_{\Tree_{u_1}}(S) + \lambda(v,u_1) \cdot w(b_{\var_w}^{(u_{i+1})},\myvec{b})$ for~$S\in \mathcal S^{(v,1)}_{\myvec{a},\myvec{b}}$, where~$w(b_{\var_w}^{(u_{i+1})},\myvec{b})$ is the \sprob at~$u_1$.
	Further, the value in entry~$F[v,\myvec{a},\myvec{b}]$ stores the correct value, when~$G[v,t,\myvec{a},\myvec{b}]$ stores the correct value, because~$\mathcal S^{(v)}_{\myvec{a},\myvec{b}}=\mathcal S^{(v,t)}_{\myvec{a},\myvec{b}}$.
	It remains to show the correctness of the value in~$G[v,i+1,\myvec a,\myvec b]$.
	
	Now, assume as an induction hypothesis that the computation of~$F[u_j,\myvec a,\myvec b]$ and~$G[v,i,\myvec a,\myvec b]$ store the correct values, for an internal vertex~$v$ with children~$u_1,\dots,u_t$ and~$i\in [t-1]$.
	We first prove that if~$G[v,i+1,\myvec{a},\myvec{b}]=p$, then there exists a set~$S\in \mathcal S^{(v,i+1)}_{\myvec{a},\myvec{b}}$ with~$PD_{\Tree_{v}}(S)=p$.
	Afterward, we prove that~$G[v,i+1,\myvec{a},\myvec{b}]\ge PD_{\Tree_{v,i+1}}(S)$ for every set~$S\in \mathcal S^{(v,i+1)}_{\myvec{a},\myvec{b}}$.
	
	($\Rightarrow$)
	Let~$G'[v,i+1,\myvec{a},\myvec{b}]=d$.
	Let~\myvec{a'} and~\myvec{b'} be the vectors that maximize the right side of Equation~(\ref{eqn:varc+varw-G-i+1}) for~$G[v,i+1,\myvec{a},\myvec{b}]$.
	By the induction hypothesis, there is a set~$S_G\in S^{(v,i)}_{\myvec{a}-\myvec{a'},\myvec{b}-\myvec{b'}}$ such that~$G[v,i,\myvec{a}-\myvec{a'},\myvec{b}-\myvec{b'}]=PD_{\Tree_{v,i}}(S_G)$ and there is a set~$S_F\in S^{(u_{i+1})}_{\myvec{a'},\myvec{b'}}$ such that~$F[u_{i+1},i,\myvec{a'},\myvec{b'}]=PD_{\Tree_{u_{i+1}}}(S_F)$.
	Define~$S := S_G \cup S_F$.
	Then,
	\begin{eqnarray}
		PD_{\Tree_{v,i+1}}(S) &=& PD_{\Tree_{v,i}}(S_G) + PD_{\Tree_{v}}(S_F)\\
		\label{eqn:varc+varw-RA2}
		&=& PD_{\Tree_v}(S_G) + PD_{\Tree_{u_{i+1}}}(S_F) + \lambda(v,u_{i+1}) \cdot w(b_{\var_w}^{(u_{i+1})},\myvec{b}).
	\end{eqnarray}
	This equals the right side of Equation~(\ref{eqn:varc+varw-G-i+1}) and we conclude~$PD_{\Tree_v}(S)=p$.
	
	($\Leftarrow$)
	Let~$S\in \mathcal \mathcal S^{(v,i+1)}_{\myvec{a},\myvec{b}}$.
	Let~$S_F$ be the set of projects~$p$ of~$S$ that are from a project list of an offspring of~$u_{i+1}$ and define~$S_G = S\setminus S_F$.
	Let~$a_k$ be the number of projects in~$S_F$ with cost~$c_k$ and let~$b_\ell$ be the number of projects in~$S_F$ with \sprob~$b_\ell$. Define~$\myvec{a'}=(a_1,\dots,a_{\var_c-1})$ and~$\myvec{b'}=(b_1,\dots,b_{\var_w-1})$.
	Then,
	\begin{eqnarray}
		\label{eqn:varc+varw-LA1}
		G[v,i+1,\myvec a,\myvec b] &\ge & G[v,i,\myvec a-\myvec{a'},\myvec b-\myvec{b'}] + F[u_{i+1},\myvec{a'},\myvec{b'}]\\
		\nonumber
		&& + \lambda(v,u_{i+1}) \cdot w\left({b'}_{\var_w}^{(u_{i+1})},\myvec{b'}\right)\\
		\label{eqn:varc+varw-LA2}
		&=& PD_{\Tree_v}(S_G) + PD_{\Tree_{u_{i+1}}}(S_F) + \lambda(v,u_{i+1}) \cdot w\left({b'}_{\var_w}^{(u_{i+1})},\myvec{b'}\right)\\
		\label{eqn:varc+varw-LA3}
		&=& PD_{\Tree_v}(S)
	\end{eqnarray}
	Inequality~(\ref{eqn:varc+varw-LA1}) follows from the recurrence in Equation~(\ref{eqn:varc+varw-G-i+1}).
	By the definition of~$S_F$ and~$S_G$, Equation~(\ref{eqn:varc+varw-LA2}) is correct.
	Finally, Equation~(\ref{eqn:varc+varw-LA3}) follows from Equation~(\ref{eqn:varc+varw-RA2}).

	\pfpara{Running time}
	First, we prove how many options for vectors~$\myvec{a}$ and \myvec{b} there are:\linebreak
	Because~$\sum_{i=1}^{\var_c-1} a_i \le |X|$, we conclude that if~$a_i=|X|$, then~$a_j=0$ for~$i\ne j$.
	Otherwise, for~$a_i \in [|X|-1]_0$ there are~$\Oh(|X|^{\var_c-1})$ options for \myvec a of not containing~$|X|$, such that there are altogether~$\Oh(|X|^{\var_c-1}+|X|)=\Oh(|X|^{\var_c-1})$ options for a suitable~\myvec a. Likewise, there are~$\Oh(|X|^{\var_w-1})$ options for a suitable~\myvec b.
	
	Clearly, the initialization can be done in~$\Oh(\numP)$~time.
	Let~$v$ be an internal vertex and let~$\myvec{a}$ and~$\myvec{b}$ be fixed.
	For a vertex~$w\in V$, we can compute in~$\Oh(n)$~time the set~$\off(w)$.
	It follows that~$w(b_{\var_w}^{(u_{i})},\myvec{b})$ can be computed in~$\Oh(n+\var_w)$~time such that Equation~(\ref{eqn:varc+varw-G-1}) can be computed in~$\Oh(n+\var_w)$~time.
	As for~$\myvec{a}$ and \myvec{b}, there are~$\Oh\left(|X|^{\var_c+\var_w-2}\right)$ options to chose~$\myvec{a'}$ and \myvec{b'}.
	Therefore,~Equation~(\ref{eqn:varc+varw-G-i+1}) can be evaluated in $\Oh\left(|X|^{\var_c+\var_w-2}\cdot (n+\var_w)\right)$~time.
	
	Equation~(\ref{eqn:varc+varw-G-1}) has to be computed once for every internal vertex.
	Equation~(\ref{eqn:varc+varw-G-i+1}) has to be computed once for every vertex except the root.
	Thus, all entries of the tables~$F$ and~$G$ can be computed in time~$\Oh\left(|X|^{2(\var_c+\var_w-2)}\cdot (n+\var_w) + |X| \cdot \numP\right)$.
	Additionally, we need $\Oh(|X|^{\var_c+\var_w-2}\cdot (\var_c+\var_w))$~time to check whether there are~$\myvec{a}$ and~$\myvec{b}$ such that
	\begin{itemize}
		\item $F[r,\myvec{a},\myvec{b}]\ge D$,
		\item $\sum_{i=1}^{\var_c-1} a_i \le |X|$,
		\item $\sum_{i=1}^{\var_w-1} b_i \le |X|$, and
		\item $a_{\var_c}^{(r)} \cdot c_{\var_c} + \sum_{i=1}^{\var_c-1} a_i \cdot c_i \le B$.
	\end{itemize}
	Because~$\Oh(n)=\Oh(|X|)$ and~$\Oh(\numP)\le \Oh(|X| \cdot \var_w)$, the overall running time\linebreak is~$\Oh\left(|X|^{2(\var_c+\var_w-1))}\cdot (\var_c+\var_w)\right)$.
\end{proof}

For~\NAP[0]{c_i}{1}{2}, where $\var_w=2$, we get the following running time.
\begin{corollary}
	\label{cor:wcode=1-varc}
	\NAP[0]{c_i}{1}{2} can be solved in~$\Oh\left(|X|^{2(\var_c+1)}\cdot \var_c\right)$ time.
\end{corollary}

As each project with a cost higher than~$B$ can be deleted, we may assume that there are no such projects which implies that~$\var_c\le C+1\le B+1$.
Thus, Theorem~\ref{thm:varc+varw} also implies that~\GNAP{} is \XP with respect to~$C+\var_w$ and~$B+\var_w$ with astronomical running times of $\Oh\left(|X|^{2(C+\var_w-1)}\cdot (C+\var_w)\right)$ and $\Oh\left(|X|^{2(B+\var_w-1)}\cdot (B+\var_w)\right)$, respectively.
However, we can adjust the algorithm so that~$B$ is not in the exponent of the running time.
Instead of declaring how many projects of cost $c_i$ for $i\in [\var_c]$ are selected, we declare the budget that can be spent.

\begin{theorem}
	\label{thm:B+varw}
	\GNAP can be solved in~$\Oh\left(B^2\cdot |X|^{2(\var_w-1)}\cdot \var_w\right)$ time.
\end{theorem}
\begin{proof}
	We describe a dynamic programming algorithm with two tables~$F$ and~$G$ that have a dimension for all the~$\var_w$ different \sprobs, except for~$\var_w-1$.
	For a vertex~$v\in V$, a given vector~\myvec{b} and $k\in [B]_0$, we define~$\mathcal S^{(v)}_{k,\myvec{b}}$ to be the family of sets of projects~$S$ such that
	\begin{enumerate}
		\item $S$ contains exactly one project of~$P_i$ for each~$x_i\in \off(v)$,
		\item $\Costs(S)\le k$, and
		\item $S$ contains exactly~$b_\ell$ projects of \sprob~$w_\ell$ for each~$\ell \in [\var_w-1]$.
	\end{enumerate}
	For a vertex~$v\in V$ with children~$u_1,\dots,u_t$, a given vector~\myvec{b} and integers~$k\in [B]_0$ and~$i\in[t]$ we define~$\mathcal S^{(v,i)}_{k,\myvec{b}}$ analogously, just that exactly one project of $P_i$ is chosen for each~$x_i\in \off(u_1)\cup\dots\cup\off(u_i)$.

	It follows that we can compute how many projects with \sprob~$w_{\var_w}$ a set~$S\in \mathcal S^{(v)}_{\myvec{a},\myvec{b}}$ contains. That are~$b_{\var_w}^{(v)} := |\off(v)|-\sum_{j=1}^{\var_w-1} b_j$ projects with \sprob~$w_{\var_w}$.
	The entries~$F[v,k,\myvec{b}]$ and~$G[v,i,k,\myvec{b}]$ store the expected phylogenetic diversity of the tree $T_v$ for~$S\in \mathcal S^{(v)}_{k,\myvec{b}}$ and $T_{v,i}$ for~$S\in\mathcal S^{(v,i)}_{k,\myvec{b}}$, respectively.
	We further define the total \sprob to be~$w(b_{\var_w},\myvec{b}) := 1-
	(1-w_{\var_w})^{b_{\var_w}}
	\cdot
	\prod_{i=1}^{\var_w-1}
	(1-w_i)^{b_i}$,
	when~$b_{\var_w}$ and~$\myvec{b}$ describe the number of chosen single~\sprobs.

	Fix a taxon~$x_i$ with project list~$P_i$.
	We store~$F[x_i,k,\myvec{b}] = 0$, if~$P_i$ contains a project~$p = (c_t,w_\ell)$ such that $c_t\le k$, and
	\begin{enumerate}
		\item ($\ell<\var_w$ and~$\myvec{b}={\myvec{0}}_{(\ell)+1}$ or~$\ell=\var_c$ and~$\myvec{b}={\myvec{0}}$).
	\end{enumerate}
	Otherwise, store~$F[x_i,k,\myvec{b}] = -\infty$.

	Let~$v$ be an internal vertex with children~$u_1,\dots,u_t$, we define
	\begin{eqnarray}
		\label{eqn:B+varw-G-1}
		G[v,1,k,\myvec{b}]
		&=&
		F[u_1,k,\myvec{b}]
		+ \lambda(v,u_1) \cdot
		w\left(b_{\var_w}^{(u_{1})},\myvec{b}\right)
	\end{eqnarray}
	and to compute further values of~$G$, we can use the recurrence
	\begin{eqnarray}
		\label{eqn:B+varw-G-i+1}
		G[v,i+1,k,\myvec{b}]
		&=&
		\max_{
			\begin{array}{l}
				k' \in [k]_0\\
				{\myvec{0}}\le \myvec{b'}\le \myvec{b}
			\end{array}}
		\left\{
		\begin{array}{l}
			G[v,i,k-k',\myvec{b}-\myvec{b'}]
			+ F[u_{i+1},k',\myvec{b'}]\\
			+ \lambda(v,u_{i+1}) \cdot
			w\left(b_{\var_w}^{(u_{i+1})},\myvec{b'}\right)
		\end{array}
		\right.
		.
	\end{eqnarray}
	Herein, we write~$\myvec{p}\le \myvec{q}$ if~$\myvec{p}$ and~$\myvec{q}$ have the same dimension~$d$ and~$p_i\le q_i$ for every~$i\in[d]$.
	And finally, we define~$F[v,k,\myvec{b}] = G[v,t,k,\myvec{b}]$.
	
	Return yes if there are~$k\in [B]_0$ and~$\myvec{b}$ such that~$\sum_{i=1}^{\var_w-1} b_i \le |X|$, and~$F[r,k,\myvec{b}]\ge D$ for the root~$r$ of~$\Tree$. Otherwise, return no.
	
	The correctness and the running time can be proven analogously to the proof of Theorem~\ref{thm:varc+varw}.
\end{proof}

Generally, we have to assume that~$B$ is exponential in the input size. However, there are special cases of~\GNAP, in which this is not the case, such as the case with unit costs for projects.
Since~$\var_w\le 2^{\wcode}$, we conclude the following from Theorem~\ref{thm:B+varw}.
\begin{corollary}
	\label{cor:C=1-XP-varw}
	\GNAP is \XP with respect to~$\var_w$ and~$\wcode$, if~$B$ is bounded polynomially in the input size.
\end{corollary}

\subsection{Generalized Noah's Ark Problem on Stars}
We now consider~\SGNAP, the special case of \GNAP where the phylogenetic~\mbox{$X$-tree}~$\Tree$ has height~1.
We first show that \SGNAP is \W1-hard with respect to the number~$|X|$ of taxa. This implies that Proposition~\ref{pps:BruteForce}\ref{pps:GNAP-X-XP} cannot be improved to an FPT algorithm.   
Afterward, we prove that most of the \FPT and \XP algorithms we presented for \MCKP can be adopted for~\SGNAP, yielding algorithms with a faster running time than for \GNAP.

\subsubsection{Hardness}
\begin{theorem}
	\label{thm:X-W1hard}
	\SGNAP is \W1-hard with respect to~$|X|$, even if the given~$X$-tree~$\Tree$ is ultrametric and~$\val_\lambda=D=1$.
\end{theorem}
\begin{proof}
	\pfpara{Reduction}
	We reduce from \MCKP, which by Theorem~\ref{thm:MCKP-m-W1} is \W1-hard with respect to the number of classes~$m$.
	Let~$\mathcal I=(N,\{N_1,\dots,N_m\},c,d,B,D)$ be an instance of \MCKP.
	We define an instance~\mbox{$\mathcal I'=(\Tree,\mathcal P,B':=B,D':=1)$} in which the~$X$-tree~$\Tree=(V,E,\lambda)$ is a star with center~$v$ and the vertex set is~$V:=\{v\} \cup X$, with~$X:=\{x_1,\dots,x_m\}$. Set~$\lambda(e):=1$ for every~$e\in E$.
	For every class~$N_i=\{ a_{i,1}, \dots, a_{i,\ell_i} \}$, define a project list~$P_i$ with projects~$p_{i,j} := (c_{i,j} := c(a_{i,j}), w_{i,j} := d(a_{i,j})/D )$. The~$|X|$-collection of project lists~$\mathcal P$ contains all these project lists~$P_i$.
	
	\pfpara{Correctness}
	Because we may assume that~$0\le d(a)\le D$ for all~$a\in N$, the \sprobs fulfill~$w_{i,j}\in [0,1]$ for all~$i\in [m]$ and $j\in [|N_i|]$.
	The tree has $m$ taxa and the reduction is clearly computable in polynomial time, so to prove that it is a parameterized reduction, it only remains to show the equivalence.
	
	($\Rightarrow$)
	Let~$S$ be a solution for~\Instance with~$S\cap N_i=\{a_{i,j_i}\}$.
	We show that~$S'=\{p_{i,j_i} \mid i\in [m]\}$ is a solution for~$\mathcal I'$:
	The cost of the set~$S'$ is~$\sum_{i=1}^m c_{i,j_i} = \sum_{i=1}^m c(a_{i,j_i}) \le B$ and further~$PD_\Tree(S') = \sum_{(v,x_i)\in E} \lambda(v,x_i) \cdot w_{i,j_i} = \sum_{(v,x_i)\in E} 1 \cdot d(a_{i,j})/D = \frac 1D \cdot \sum_{i=1}^m d(a_{i,j}) \ge 1 = D'$.
	
	($\Leftarrow$)
	Let~$S=\{p_{1,i_1},\dots,p_{m,j_{m}}\}$ be a solution for~$\Instance'$.
	We show that~$S'=\{a_{1,i_1},\dots,a_{m,j_{m}}\}$\linebreak is a solution for \Instance:
	Clearly,~$S'$ contains exactly one item per class.
	The cost\linebreak of the set~$S'$ is~$c_\Sigma(S') = \sum_{i=1}^m c(a_{i,j_i}) = \sum_{i=1}^m c_{i,j_i} \le B$.
	The diversity of~$S'$\linebreak is~\mbox{$d_\Sigma(S') = \sum_{i=1}^m d(a_{i,j_i}) = \sum_{i=1}^m w_{i,j_i} \cdot D = PD_\Tree(S)\cdot D \ge D$}.
\end{proof}

The~$X$-tree constructed in the reduction above is a star and therefore has unbounded maximum degree~$\Delta$.
In the following, we show that~\GNAP also is \W{1}-hard with respect to~$|X|$ when~$\Tree$ is a binary tree.
\begin{corollary}
	\label{cor:X+Delta}
	\GNAP is \W1-hard with respect to~$|X|+D$ even if~$\Delta=3$ and~$\val_\lambda=1$.
\end{corollary}
\begin{proof}
	\pfpara{Reduction}
	We reduce from~\SGNAP, which by Theorem~\ref{thm:X-W1hard} is \W{1}-hard with respect to~$|X|$, even if~$\val_\lambda=D=1$.
	Let~$\Instance=(\Tree,\mathcal P,B,D)$ be an instance of~\GNAP with $D=1$.
	Define an~$X$-tree~$\Tree':=(V,E)$ as follows:
	Let~$V:=X \cup\{v_1,\dots,v_{|X|},x^*\}$ and~$E:= \{(v_i,x_i),(v_i,v_{i+1}) \mid i\in [|X|-1]\} \cup \{ (v_n,x_n),(v_n,x^*) \}$, and set~$\lambda(e):=1$ for each edge~$e$ of~$\Tree'$.
	Define a project-list~$P_{x^*} = (p_{*,0} := (0,0),p_{*,1} := (1,1))$ for~$x^*$.
	Finally, let~$\Instance'$ be~$(\Tree',\mathcal P\cup P_{x^*},B':=B+1,D':=|X|+1)$.
	
	\pfpara{Correctness}
	The reduction can be computed in polynomial time and the size of the taxa set has increased by only one. Note that we can assume that a solution of~$\Instance'$ contains~$p_{*,1}$ because otherwise exchanging an arbitrary project with~$p_{*,1}$ gives a better solution. It thus remains only to show that~$S$ is a solution for~\Instance if and only if~$S':=S\cup\{ p_{*,1} \}$ is a solution for~$\Instance'$.
	Clearly,~$\Costs(S') = \sum_{p_{i,j} \in S'} c_{i,j} = c_{*,1} + \sum_{p_{i,j} \in S} c_{i,j} = \Costs(S)+1$.
	Because~$w_{*,1}=1$, we conclude that the \sprob at each vertex~$v_i$ is exactly~1. Thus, the value of~$PD_{\Tree'}(S')$ is
	$$
	\sum_{i=1}^{|X|-1} \lambda(v_i,v_{i+1}) \cdot 1
	+
	\lambda(v_{|X|},x^*) \cdot 1
	+
	\sum_{p_{i,j}\in S} \lambda(v_i,x_i) \cdot w_{i,j} =
	|X| + PD_{\Tree}(S).
	$$
	Hence,~$\Costs(S') \le B+1$ and~$PD_{\Tree'}(S')\ge |X|+1$ if and only if~$\Costs(S) \le B$\linebreak and~$PD_{\Tree}(S)\ge 1$.
\end{proof}

%\kommentar{\subsubsection{Algorithmic results}}
In Section~\ref{sec:MCKP}, we presented a set of algorithms for~\MCKP. We will now make use of these algorithms to obtain several tractability results for \SGNAP. 

\begin{proposition}
	\label{pps:height=1->mckp}
	\SGNAP can be solved in time
	\begin{itemize}
		\item $\Oh(D\cdot 2^{\wcode} \cdot \numP + |\mathcal{I}|)$, 
		\item $\Oh(B \cdot \numP + |\mathcal{I}|)$, 
		\item $\Oh(C \cdot \numP \cdot |X| + |\mathcal{I}|)$, and
		\item $\Oh(|X|^{\var_c-1} \cdot \numP + |\mathcal{I}|)$.
	\end{itemize}
\end{proposition}
\begin{proof}
	To show the statement, we describe how to reduce~\SGNAP  to~\MCKP and then use algorithms presented in Section~\ref{subsec:algos-MCKP}.
	
	\pfpara{Reduction}
	Let~$\mathcal{I}=(\Tree,\lambda,\mathcal{P},B,D)$ be an instance of~\SGNAP.
	We define an instance~$\mathcal{I}'=(N,\{N_1,\dots,N_{|X|}\},c,d,B',D')$ of~\MCKP.
	Without loss of generality, each \sprob is in the form~$w_i=w_i'/2^{\wcode}$ with~$w_i'\in \mathbb{N}_0$, and~$w_i'\le 2^{\wcode}$.
	For every taxon~$x_i$ with project list~$P_i$, define the class~$N_i$ and for every project~$p_{i,j}=(c_{i,j},w_{i,j})\in P_i$ add an item~$a_{i,j}$ to~$N_i$ with cost~$c(a_{i,j}):=c_{i,j}$ and value~$d(a_{i,j}):=w_{i,j}'\cdot \lambda(w,x_i)$. Set~$B':=B$ and~$D':=D\cdot 2^{\wcode}$.
	
	\pfpara{Correctness}
	Let~$S=\{p_{1,j_1},\dots,p_{|X|,j_{|X|}}\}$ be a solution for the instance~\Instance of~\GNAP.
	Define the set~$S'=\{a_{1,j_1},\dots,a_{|X|,j_{|X|}}\}$. Clearly, \mbox{$c_\Sigma(S')=\Costs(S)\le B$}.
	Further,
	\begin{eqnarray*}
		\sum_{i=1}^{|X|} d(a_{i,j_i}) &=& \sum_{i=1}^{|X|} w_{i,j_i}'\cdot \lambda(w,x_i)\\
		&=& 2^{\wcode}\cdot \sum_{i=1}^{|X|} w_{i,j_i}\cdot \lambda(w,x_i)\\
		&=& 2^{\wcode}\cdot PD_\Tree(S) \ge 2^{\wcode}\cdot D = D'.
	\end{eqnarray*}
	Thus,~$S'$ is a solution for~$\Instance'$.
	
	Analogously, one can show that if~$S'$ is a solution for the instance~$\Instance'$ of~\MCKP, then~$S$ is a solution for the~\GNAP-instance~\Instance.
	
	\pfpara{Running time}
	The reduction can be computed in~$\Oh(|\Instance|)$ time.
	We observe that in~$\Instance'$ the size of the set~$N$ equals the number~$\numP$ of projects, the number of classes~$m$ is the number~$|X|$ of taxa, and the budget~$B$ remains the same. Because all costs are simply copied, the maximal cost~$C$ and the number of different costs~$\var_c$ remain the same.
	Because the \sprobs are multiplied with an edge weight, it follows that~$\var_d \in \Oh(\var_w \cdot \val_\lambda)$. By definition,~$D' = D\cdot 2^{\wcode}$.
	
	Thus, with this reduction at hand, we can obtain any of the claimed running times by using the \MCKP running time bound given in Table~\ref{tab:results-mckp} for the corresponding parameter.
\end{proof}
To obtain an FPT algorithm for $\var_c+\var_d$ we cannot directly make use of the corresponding FPT result for \MCKP because $\var_d$ may be as large as $\var_w\cdot \val_\lambda$. Instead, we present a direct reduction from \SGNAP to~\ILPF. 
\begin{theorem}
	\label{thm:height=1:varc+varw}
	There is a reduction from~\SGNAP  to~\ILPF with~$\Oh(2^{\var_c+\var_d}\cdot \var_c)$ variables.
	Thus, \SGNAP is \FPT with respect to~$\var_c+\var_w$.
\end{theorem}
\begin{proof}
	\pfpara{Reduction}
	Let~$\mathcal{I}=(\Tree,\lambda,\mathcal{P},B,D)$ be an instance of~\SGNAP and let~$\rho$ denote the root of~$\Tree$.
	
	We may assume that a project list~$P_i$ does not contain two projects of the same cost or the same value.
	In the following, we call~$T=(C,W)$ a \textit{type}, for sets~$C\subseteq \mathcal C$ and~$W\subseteq \mathcal W$ with~$|C| = |W|$, where~$\mathcal C$ and~$\mathcal W$ are the sets of different costs and~\sprobs, respectively. Let~$\mathcal X$ be the family of all types.
	We say that the \textit{project list~$P_i$ is of type~$T=(C,W)$} if~$C$ and~$W$ are the set of costs and \sprobs of~$P_i$.
	For each~$T\in\mathcal X$, we define~$m_T$ be the number of classes of type~$T$.
	
	Observe, for each type~$T=(C,W)$, project list~$P_i$ of type~$T$, and a project~$p\in P_i$, we can compute the \sprob of~$p$ when we know the cost~$c_k$ of~$p$. More precisely, if~$c_k$ is the~$\ell$th cheapest cost in~$C$, then the \sprob of~$p$ is the~$\ell$th smallest \sprob in~$W$.
	For a type~$T=(C,W)$ and~$i\in [\var_c]$, we define the constant~$w_{T,i}$ to be~$-n\cdot \val_\lambda$ if~$c_i \not\in C$. Otherwise, let~$w_{T,i}\in [0,1]$ be the~$\ell$th smallest \sprob in~$W$, if~$c_i$ is the~$\ell$th smallest cost in~$C$.
	
	For two taxa~$x_i$ and~$x_j$ with project lists~$P_i$ and~$P_j$ of the same type~$T$, it is possible that~$\lambda(v,x_i)\ne\lambda(v,x_j)$.
	Hence, it can make a difference if a project is selected for the taxon~$x_i$ instead of~$x_j$.
	For a type~$T$, let~$x_{T,1},\dots,x_{T,m_T}$ be the taxa, such that the project lists~$P_{T,1},\dots,P_{T,m_T}$ are of type~$T$ and let~$\lambda(v,x_{T,i})\ge\lambda(v,x_{T,i+1})$ for each~$i\in [m_T -1]$.
	For each type $T$, we define a function~$f_T: [m_T]_0 \to \mathbb{N}$ by~$f_T(0):=0$ and~$f_T(\ell)$ stores total value of the first~$\ell$ edges. More precisely, that is~$f_T(\ell) := \sum_{i=1}^{\ell} \lambda(v,x_i)$.
	
	The following describes an instance of~\ILPF.
	\begin{align}
		\label{eqn:ILPF-B}
		\sum_{T \in \mathcal{X}} \sum_{i=1}^{\var_c} y_{T,i} \cdot c_i \le & \; B\\
		\label{eqn:ILPF-D}
		\sum_{T \in \mathcal{X}} \sum_{i=1}^{\var_c} w_{T,i} \cdot g_{T,i} \ge & \; D\\
		\label{eqn:ILPF-g}
		f_{T}\left(\sum_{\ell=i}^{\var_c} y_{T,\ell}\right) - f_{T}\left(\sum_{\ell=i+1}^{\var_c} y_{T,\ell}\right) = & \; g_{T,i} & \quad \forall T\in\mathcal{X}, i\in [\var_c]\\
		\label{eqn:ILPF-Nr}
		\sum_{i=1}^{\var_c} y_{T,i} = & \; m_{T} & \quad \forall T\in \mathcal{X}\\
		\label{eqn:ILPF-variables}
		y_{T,i}, g_{T,i} \ge & \; 0 & \quad \forall T\in\mathcal{X}, i\in [\var_c]
	\end{align}
	The variable~$y_{T,i}$ expresses the number of projects with cost~$c_i$ that are chosen in a project list of type~$T$.
	We want to assign the most valuable edges that are incident with taxa that have a project list of type~$T$ to the taxa in which the highest \sprob is chosen. To receive an overview, in~$g_{T,\var_c}$ we store the total value of the~$y_{T,\var_c}$ most valuable edges that are incident with a taxon that has a project list of type~$T$. Then, in~$g_{T,\var_c-1}$ we store the total value of the next~$y_{T,\var_c-1}$ most valuable edges, and so on.
	
	For each type~$T$, the function~$f_T$ is not necessarily linear. However, for each~$f_T$ there are affine linear functions~$p_T^{(1)},\dots,p_T^{(m_T)}$ such that~$f_T(i) = \min_\ell p_T^{(\ell)}(i)$ for each~$i\in [m_T]$~\cite{etscheid}.
	
	\pfpara{Correctness}
	Observe that if~$c_i\not\in C$, then because we defined~$d_{T,i}$ to be~$-n\cdot\val_\lambda$, Inequality~(\ref{eqn:ILPF-D}) would not be fulfilled if~$g_{T,i} > g_{T,i+1}$ and consequently $y_{T,i}=0$ if~$c_i\not\in C$ for each type~$T=(C,W)\in\mathcal X$ and~$i\in[\var_c]$.
	Inequality~(\ref{eqn:ILPF-B}) ensures that the total cost is at most~$B$.
	Inequality~(\ref{eqn:ILPF-D}) ensures that the total phylogenetic diversity is at least~$D$.
	The variable~$g_{T,i}$ stores  the total weight of the edges towards the~$y_{T,i}$ taxa with projects of project lists of type~$T$, in which a project of cost~$c_i$ is selected. All these projects have \sprob~$w_{T,i}$, and thus the phylogenetic diversity of these projects is~$w_{T,i} \cdot (g_{T,i} - g_{T,i+1})$.   
	Equation~(\ref{eqn:ILPF-g}) ensures that the value of~$g_{T,i}$ is chosen correctly.
	Equation~(\ref{eqn:ILPF-Nr}) ensures that  exactly~$m_T$ projects are picked from the project lists of type~$T$, for each~$T\in \mathcal X$. Altogether, this shows the correctness of the reduction.
	Finally, the constructed instance has~$\Oh(2^{\var_c+\var_d}\cdot \var_c)$ variables because~$\mathcal X\subseteq 2^{\mathcal C} \times 2^{\mathcal W}$ which implies that there are~$\Oh(2^{\var_c+\var_d}\cdot \var_c)$ different options for the variables~$y_{T,i}$ and~$g_{T,i}$.
\end{proof}

\section{Restriction to Two Projects per Taxon}\label{sec:two-projects}
We finally study two special cases of~\NAP[a_i]{c_i}{b_i}{2}---the special case of~\GNAP where every project list contains exactly two projects.

\subsection{Sure Survival or Extinction for each Project}
First, we consider~\NAP[0]{c_i}{1}{2}, the special case where each taxon~$x_i$  survives definitely if cost~$c_i$ is paid and  becomes extinct, otherwise.
This special case was introduced by Pardi and Goldman~\cite{pardi} who also presented a pseudopolynomial-time algorithm that computes a solution in~$\Oh(B^2 \cdot n)$~time.
Because we may assume that~$B\le C\cdot \bet X$, we conclude the following.
\begin{corollary}
	\label{cor:C-01-NAP}
	\NAP[0]{c_i}{1}{2} can be solved in~$\Oh(C^2 \cdot {\bet X}^3)$ time.
\end{corollary}

We now show that~\NAP[0]{c_i}{1}{2} is \FPT with respect to~$D$, with an adaption of the above-mentioned algorithm of Pardi and Goldman~\cite{pardi} for the parameter~$B$.
\begin{proposition}
	\label{pps:wcode=1-D}
	\NAP[0]{c_i}{1}{2} can be solved in~$\Oh(D^2 \cdot |X|)$ time.
\end{proposition}
\begin{proof}
	\pfpara{Table definition}
	We write~$c(x_i)$ for the cost of the project with \sprob~1 in~$P_i$.
	For a set~$A$ of vertices, a vertex~$v$, and integers~$b$ and~$d$, we call a set of projects~$S$ an~\textit{$(A,v,d)$-respecting set}, if~$PD_{\Tree_v}(S)\ge d$, and~$S$ contains exactly one project of the project lists of the offspring of~$A$, and~$S$ contains at least one project with \sprob~1.
	
	We describe a dynamic programming algorithm with two tables~$F$ and~$G$.
	We want entry~$F[v,d]$ for a vertex~$v\in V$, an integer~$d \in [D]_0$ to store the minimal cost of a~$(\{v\},v,d)$-respecting set.
	If~$v$ is an internal vertex with children~$u_1,\dots,u_t$ and~$i\in [t]$,
	then we want entry~$G[v,i,d]$ to store the minimal cost of a~$(\{u_1,\dots,u_i\},v,d)$-respecting~set.

	\pfpara{Algorithm} In the algorithm description, we let~$\mo$ denote non-negative subtraction, that is, $a \mo b := \max(a - b,0)$.

	As basic cases for a leaf~$x_i$ store~$F[x_i,0] = c(x_i)$ and for each~$d > 0$ store~$F[x_i,d] = \infty$.

        Now, let~$v$ be an internal vertex with children~$u_1,\dots,u_t$. We define~$G[v,1,d] = F[u_1,d \mo \lambda(v u_1)]$.
	If for a fixed~$i\in [t]$ the values of~$G[v,i,d]$ and~$F[u_{i+1},d]$ are known for each~$d\in [D]_0$, we compute the value of~$G[v,i+1,d]$ with the recurrence
	\begin{equation}
		\label{eqn:wcode=1-D}
		G[v,i+1,d]
		 =  \min \{ G[v,i,d]; \min_{d'\in [d_{i+1}]} \{ G[v,i,d_{i+1}-d'] + F[u_{i+1},d'] \} \},
		\nonumber
	\end{equation}
	where~$d_{i+1} := d \mo \lambda(v u_{i+1})$.
	We set~$F[v,d]:=G[v,t,d]$, eventually.
	
	We return yes if~$F[\rho,D] \le B$ for the root~$\rho$ of~$\Tree$ and no, otherwise.
	
	\pfpara{Correctness}
	The basic cases, the computation of~$G[v,1,d]$, and the computation of~$F[v,d]$ for an internal vertex~$v$ and~$d\in [D]_0$ are correct by definition.
	It remains to show that~$G[v,i+1,d]$ stores the correct value if~$G[v,i,d]$ and~$F[u_{i+1},d]$ do.
	We first show that if~$S$ is\linebreak a~$(\{u_1,\dots,u_{i+1}\},v,d)$-respecting set, then~$G[v,i+1,d]\le \Costs(S)$.
	Afterward, we show that if~$G[v,i+1,d]=c$, then there is a~$(\{u_1,\dots,u_{i+1}\},v,d)$-respecting set~$S$ with~$\Costs(S)=c$.
	
	($\Rightarrow$)
	Let~$S$ be a~$(\{u_1,\dots,u_{i+1}\},v,d)$-respecting set.
	Let~$S_1$ be the set of projects that are in one of the project lists of the offspring of~$u_{i+1}$.
	Define~$S_2 := S\setminus S_1$.
	If~$S_2$ only contains projects of \sprob~0, then $G[v,i+1,d] = G[v,i,d]$.
	Otherwise, let~$d' := PD_{\Tree_v}(S_1)$.
	Then,~$PD_{\Tree_v}(S_2) = PD_{\Tree_v}(S) - PD_{\Tree_v}(S_1) \ge d - d'$.
	We conclude that~$S_1$ is a~$(\{u_{i+1}\},v,d')$-respecting set and~$S_2$ is a~$(\{u_1,\dots,u_{i}\},v,d-d')$-respecting set.
	Consequently,
	\begin{eqnarray}
		\label{eqn:wcode=1-LA-1}
		G[v,i+1,d] &\le& G[v,i,d-d'] + F[u_{i+1},d' \mo \lambda(v u_{i+1})]\\
		\label{eqn:wcode=1-LA-2}
		&\le& \Costs(S_G) + \Costs(S_F) = \Costs(S).
	\end{eqnarray}
	Here, Inequality~(\ref{eqn:wcode=1-LA-1}) follows from the recurrence in Recurrence~(\ref{eqn:wcode=1-D}) and Inequality~(\ref{eqn:wcode=1-LA-2}) follows from the induction hypothesis.
	
	($\Leftarrow$)
	Let~$G[v,i+1,d]$ store~$c$.
	Unless~$G[v,i+1,d] = G[v,i,d]$ there is an integer~$d'\in [d_{i+1}]_0$, such that~$G[v,i+1,d] = G[v,_{i+1},d_{i+1}-d'] + F[u_{i+1},d']$.
	By the induction hypothesis, there is a~$(\{u_{i+1}\},v,d')$-respecting set~$S_1$ and a~$(\{u_1,\dots,u_{i}\},v,d-d')$-respecting set~$S_2$ such that~$F[u_{i+1},d']=\Costs(S_1)$ and~$G[v,i,d-d']=\Costs(S_2)$.
	We conclude that~$S:=S_1\cup S_2$ is a~$(\{u_1,\dots,u_{i+1}\},v,d)$-respecting set and~$c = \Costs(S_1)+\Costs(S_2)=\Costs(S)$.

	\pfpara{Running time}
	Table~$F$ has~$\Oh(D \cdot n)$ entries, and each entry can be computed in constant time.
	Also,~$G$ has~$\Oh(D \cdot n)$ entries.
	Entry~$G[v,i+1,d]$ is computed by checking at most~$D+1$ options for~$d'$.
	Altogether, a solution can be found in~$\Oh(D^2 \cdot n)$~time.
\end{proof}

We may further use this pseudopolynomial-time algorithm to obtain an algorithm for the maximum edge weight~$\val_\lambda$:
If we are given an input instance~\Instance of~\NAP[0]{c_i}{1}{2} such that~$\sum_{e\in E} \lambda(e) < D$, then we may directly return no since the desired diversity can never be reached. After this check, we may assume~$D \le \sum_{e\in E} \lambda(e)\le \val_\lambda \cdot (n-1)$. This gives the following running time bound. 
\begin{corollary}
	\label{cor:wcode=1,val-lambda}
	\NAP[0]{c_i}{1}{2} can be solved in~$\Oh((\val_\lambda)^2 \cdot |X|^3)$ time.
\end{corollary}

\subsection{Unit Costs for each Project}
\label{subsec:unitc}
Second, we consider~\unitc---the special case of \GNAP in which every project with a positive \sprob has the same cost.
In the following we use the term \emph{solution} to denote only those projects of cost 1 which have been chosen.
Further, with~$w(x_i)$ we denote the \sprob of the project in~$P_i$ which costs~1.

We consider the following basic problem.
\Problem{Penalty-Sum}
{A set of tuples~$T=\{ t_i=(a_i,b_i) \mid i\in [n], a_i\in \mathcal{Q}_{\ge 0},b_i\in (0,1) \}$,
	two integers~$k,Q$,
	a number~$D\in \mathbb{Q}_+$}
{Is there a subset~$S\subseteq T$ of size~$k$ such that
	$\sum_{t_i\in S} a_i - Q\cdot \prod_{t_i\in S} b_i \ge D$}

In an earlier version of this work~\cite{GNAP-CTW}, we asked for complexity of \PS. It has been proven meanwhile that \PS is NP-hard~\cite{MaxNPD}.
We now present a polynomial-time many-to-one reduction from~\PS to~\unitc, giving the following hardness result.
\begin{theorem}
	\label{thm:PS-hardness}
	\unitc is NP-hard, even on~$X$-trees with height at most~2 and a root~$r$ of degree~1.
\end{theorem}
\begin{proof}
	\pfpara{Reduction}
	Let~$\Instance = (T,k,Q,D)$ be an instance of~\PS.
	Let~$\bin(a)$ and~$\bin(1-b)$ be the maximum binary encoding length of any~$a_i$ and~$1-b_i$, respectively, and let~$W := 2^{\bin(a)+\bin(1-b)}$.
	We define an instance~$\Instance' = (\Tree,\mathcal{P},B,D')$ of~\unitc as follows:
	Let~$\Tree$ contain the vertices~$V:=\{r,v,x_1,\dots,x_{|T|}\}$ and let~$\Tree$ be a star with center~$v$ and root~$r$.
	We define~$\lambda(r,v)=W Q$ and~$\lambda(v,x_i) = W \cdot (\nicefrac{a_i}{1-b_i})$ for each~$t_i\in T$.
	For each tuple~$t_i$, we define a project list~$P_i := ((0,0),(1,1-b_i))$. Then,~$\mathcal P$ is defined to be the set of these project lists.
	Finally, we set~$B:=k$, and~$D':=W (D+Q)$.
	The reduction can clearly be  computed in polynomial time.
	
	\pfpara{Correctness}
	We show that~\Instance is a yes-instance of~\PS if and only if~$\Instance'$ is a yes-instance of~\unitc.

		($\Rightarrow$)
	Let~$S$ be a solution for~\Instance of~\PS. Let the set~$S'$ contain the project~$(1,w(x_i))$ if and only if~$t_i\in S$.
	Then,
	\begin{eqnarray*}
		&& \sum_{x_i\in S'} \lambda(v,x_i) \cdot w(x_i) + \lambda(r,v)\cdot \left(1-\prod_{x_i\in S'} (1- w(x_i))\right)\\
		&=& \sum_{t_i\in S'} W\cdot(\nicefrac{a_i}{1-b_i})\cdot (1-b_i) + W Q\cdot \left(1-\prod_{t_i\in S'} (1-(1-b_i))\right)\\
		&=& \sum_{t_i\in S} W a_i - W Q\cdot \prod_{t_i\in S} b_i + W Q \ge W \cdot (D + Q) = D'.
	\end{eqnarray*}
	Because~$|S'|=|S|\le B$, we have that~$S'$ is a solution for the instance~$\Instance'$ of~\unitc.
		
		($\Leftarrow$)
	Let~$S'$ be a solution for the instance~$\Instance'$ of~\unitc.
	We conclude\linebreak that~$\sum_{x_i\in S'} \lambda(v,x_i) \cdot w(x_i) \ge D' - \lambda(r,v)\cdot \left(1-\prod_{x_i\in S'} (1- w(x_i))\right)$.
	Define the set~$S\subseteq T$ to contain a tuple~$t_i$ if and only if~$S'$ contains a project of the taxon~$x_i$.
	Then,
	\begin{eqnarray*}
		&&\sum_{t_i\in S} a_i - Q\cdot \prod_{t_i\in S} b_i\\
		&=& \sum_{x_i\in S'} W^{-1}\cdot \lambda(v,x_i)\cdot (1-b_i) - W^{-1}\cdot \lambda(r,v)\cdot \prod_{x_i\in S'} (1- w(x_i))\\
		&=& W^{-1}\cdot \left(\sum_{x_i\in S'} \lambda(v,x_i)\cdot w(x_i) - \lambda(r,v)\cdot \prod_{x_i\in S'} (1- w(x_i))\right)\\
		&\ge& W^{-1}\cdot \left(D' - \lambda(r,v)\cdot \left(1-\prod_{x_i\in S'} (1- w(x_i))\right) - \lambda(r,v)\cdot \prod_{x_i\in S'} (1- w(x_i))\right)\\
		&=& W^{-1}\cdot \left(D' - \lambda(r,v)\right) = W^{-1} \cdot (W(D+Q) - W Q) = D.
	\end{eqnarray*}
	Because~$|S'|=|S|\le B$, we may conclude that~$S$ is a solution of~\PS.
\end{proof}

%%%%%%%%%%%%
%%% greedy-algo %%%
%%%%%%%%%%%%
%\kommentar{
%\todosi{We did not include the following Proposition in the CTW-version. However, I think it is correct. Do you think we can add it to arXiv or to the journal-version?}
%\begin{proposition}
%\label{pps:C=1,height=2,ultrametric}
%\unitc can be solved greedily when the given~$X$-tree~$\Tree$ is ultrametric and has a~$\height_\Tree$ of at most two.
%\end{proposition}
%\begin{proof}
%\todos[inline]{ToDo}
%\end{proof}
%}
%%%%%%%%%%%%
%%% greedy-algo %%%
%%%%%%%%%%%%

Recall that in an ultrametric tree, the weighted distance from the root to all vertices is the same.
Observe that \unitc can be solved exactly on ultrametric trees of height at most~2 by greedily selecting a taxon which adds the highest estimated diversity.
In the following theorem, we show that~\unitc is NP-hard even when restricted to ultrametric trees of height at most~3.
\begin{theorem}
	\label{thm:C=1,height=3,ultrametric}
	\unitc is NP-hard on ultrametric trees of height at most~3.
\end{theorem}
\begin{proof}
	By Theorem~\ref{thm:PS-hardness}, it is sufficient to reduce from \unitc with the restriction that the root has only one child and the height of the tree is~2.
	
	\pfpara{Reduction}
	Let~$\Instance = (\Tree,\mathcal{P},B,D)$ be an instance of~\unitc in which the root~$r$ of the~$X$-tree~$\Tree$ has only one child~$v$ and~$\height_\Tree=2$.
	Without loss of generality, assume for each~$i\in [|X|-1]$ that~$\lambda(v,x_i)\ge \lambda(v,x_{i+1})$ and there is a fixed~$s\in [|X|]$ with~$w(x_s)\ge w(x_j)$ for each~$j\in [|X|]$.
	Observe, that by the reduction in Theorem~\ref{thm:PS-hardness} we may assume~$w(x_j) \ne 1$ for each~$x_j\in X$.
	
	We define an instance~$\Instance' := (\Tree',\mathcal{P},B',D')$ of~\unitc as follows:
	Let~$X_1\subseteq X$ be the set of vertices~$x_i$ with~$\lambda(v,x_1)=\lambda(v,x_i)$. If~$X_1=X$, then~$\Instance$ is already ultrametric and the reduction may simply output~$\Instance$.
	Assume otherwise and define~$X_2 := X\setminus X_1$.
	Fix an integer~$W\in\mathbb{N}$ that is large enough such that~$1-W^{-1}>w_{s,1}$ and~$W\cdot \lambda(v,x_{|X|})>\lambda(v,x_1)$.
	Define a tree~$\Tree'=(V',E',\lambda')$, in which~$V'$ contains the vertices~$V$ and for every~$x_i\in X_2$, we add two vertices~$u_i$ and~$x_i^*$.
	Let~$X^*$ be the set that contains all~$x_i^*$.
	The set of edges is defined by~$E'=\{(r,v)\}\cup \{(v,x_i) \mid x_i\in X_1\} \cup \{ (v,u_i), (u_i,x_{i}), (u_i,x_i^*) \mid x_i\in X_2 \}$. Observe that the leaf set of~$\Tree$ is~$X\cup X^*$.
	The weights of the edges are defined by:
	\begin{itemize}
		\item $\lambda'(r,v) = \lambda(r,v) \cdot W^{|X_2|} \cdot (W - 1)$.
		\item For~$x_i\in X_1$ we set~$\lambda'(v,x_i) = (W - 1) \cdot \lambda(v,x_1)$.
		\item For~$x_i\in X_2$ we set~$\lambda'(v,u_i) = W \cdot (\lambda(v,x_1) - \lambda(v,x_i))$\\
		and~$\lambda'(u_i,x_{i}) = \lambda'(u_i,x_i^*) = (W\cdot \lambda(v,x_i)) - \lambda(v,x_1)$.
	\end{itemize}
	
	The collection~$\mathcal P'$ contains the unchanged project lists~$P_i$ for taxa~$x_i\in X$ and we add the project lists~$P_{i}^*:=((0,0),(1,1-W^{-1}))$ for each taxon~\mbox{$x_i^*\in X^*$}.
	Finally, we define~$B':=B+|X^*|$ and $$D'=(W-1)\cdot \left(D + |X_2|\cdot (W-1)\cdot \lambda(v,x_1)+\left(W^{|X_1|}-1\right)\cdot \lambda(w,v)\right).$$
	Figure~\ref{fig:reduction} depicts an example of this reduction.
\begin{figure}[t]
\centering
\begin{tikzpicture}[xscale=1.1,yscale=0.9,
		sibling distance=5em,
		every node/.style = {
			shape=circle,
			fill=white,
			rounded corners,
						minimum size=9,
						inner sep=0,
			draw,}
	]
	\tikzstyle{txt}=[circle,fill=none,draw=white,inner sep=0pt]

		\draw (-2,-1) -- (2,-1);
	
	\node (r) at (0,0) {};
	\node (v) at (0,-1) {};
	\node (x1) at (-2,-7) {};
	\node (x2) at (-1,-7) {};
	\node (x3) at (0,-5) {};
	\node (x4) at (1,-3) {};
	\node (x5) at (2,-2) {};

	\node[txt] at (-0.25,0.25) {$r$};
	\node[txt] at (-0.25,-0.75) {$v$};
	\node[txt] at (-2,-7.45) {$x_1$};
	\node[txt] at (-1,-7.45) {$x_2$};
	\node[txt] at (0,-5.45) {$x_3$};
	\node[txt] at (1,-3.45) {$x_4$};
	\node[txt] at (2,-2.45) {$x_5$};

	\node[txt] at (0.25,-0.5) {$1$};
	\node[txt] at (-2.25,-4) {$6$};
	\node[txt] at (-1.25,-4) {$6$};
	\node[txt] at (-0.25,-3) {$4$};
	\node[txt] at (0.75,-2) {$2$};
	\node[txt] at (1.75,-1.5) {$1$};

	\draw[-stealth] (r) -- (v);
	\draw[-stealth] (-2,-1) -- (x1);
	\draw[-stealth] (-1,-1) -- (x2);
	\draw[-stealth] (v) -- (x3);
	\draw[-stealth] (1,-1) -- (x4);
	\draw[-stealth] (2,-1) -- (x5);
\end{tikzpicture}
\hfill
\begin{tikzpicture}[xscale=1.6, yscale=0.9,
		sibling distance=5em,
		every node/.style = {
			shape=circle,
			fill=white,
			rounded corners,
			minimum size=9,
						inner sep=0,
			draw}
	]
	\tikzstyle{txt}=[circle,fill=none,draw=none,inner sep=0pt]

		%lines across u3,u4,u5
		\draw (-0.25,-3.286) -- (0.25,-3.286);
	\draw (0.75,-5.571) -- (1.25,-5.571);
	\draw (1.75,-6.5) -- (2.25,-6.5);
	
		\node[txt] at (1.2,-0.5) {$2^9\cdot 7=3584$};
		\draw (-2,-1) -- (2,-1);

		\node[txt] at (1.5,-6.8) {$2$};

	\node[txt] at (-0.25,0.25) {$r$};
	\node[txt] at (-0.25,-0.75) {$v$};
	\node[txt] at (-2,-7.45) {$x_1$};
	\node[txt] at (-1,-7.45) {$x_2$};
	\node[txt] at (-0.25,-7.45) {$x_3$};
	\node[txt] at (0.75,-7.45) {$x_4$};
	\node[txt] at (1.75,-7.45) {$x_5$};
	\node[txt] at (0.25,-7.45) {$x_3^*$};
	\node[txt] at (1.25,-7.45) {$x_4^*$};
	\node[txt] at (2.25,-7.45) {$x_5^*$};
	\node[txt] at (-0.25,-3.036) {$u_3$};
	\node[txt] at (0.75,-5.321) {$u_4$};
	\node[txt] at (1.75,-6.264) {$u_5$};

	\node (r) at (0,0) {};
	\node (v) at (0,-1) {};
	\node (u3) at (0,-3.286) {};
	\node (u4) at (1,-5.571) {};
	\node (u5) at (2,-6.514) {};
	\node (x1) at (-2,-7) {};
	\node (x2) at (-1,-7) {};
	\node (x3) at (-0.25,-7) {};
	\node (x4) at (0.75,-7) {};
	\node (x5) at (1.75,-7) {};
	\node (x3*) at (0.25,-7) {};
	\node (x4*) at (1.25,-7) {};
	\node (x5*) at (2.25,-7) {};

	\node[txt] at (-2.25,-4) {$42$};
	\node[txt] at (-1.25,-4) {$42$};
	\node[txt] at (-0.25,-2.1) {$16$};
	\node[txt] at (0.75,-3.3) {$32$};
	\node[txt] at (1.75,-3.9) {$40$};

	\draw[dashed,->] (r) -- (v);
	
	\draw[-stealth] (-2,-1) -- (x1);
	\draw[-stealth] (-1,-1) -- (x2);

	\draw[-stealth] (v) -- (u3);
	\draw[-stealth] (1,-1) -- (u4);
	\draw[-stealth] (2,-1) -- (u5);

	\draw[-stealth] (-0.25,-3.286) -- (x3);
	\draw[-stealth] (0.25,-3.286) -- (x3*);
	\draw[-stealth] (0.75,-5.571) -- (x4);
	\draw[-stealth] (1.25,-5.571) -- (x4*);
	\draw[-stealth] (1.75,-6.5) -- (x5);
	\draw[-stealth] (2.25,-6.5) -- (x5*);
	\node[txt] at (0,-5.1) {$26$};
	\node[txt] at (1,-6.3) {$10$};
\end{tikzpicture}
\caption{An example of the reduction presented in Theorem~\ref{thm:C=1,height=3,ultrametric}, where the left side shows an example-instance~\Instance and  the right side shows the instance~$\Instance'$. Here, the~\sprobs are omitted and we assume that~$t=3$.}
\label{fig:reduction}
\end{figure}
Observe that~$t$ can be chosen such that~$t\in \Oh(\wcode + \val_\lambda)$. Hence, the reduction can be computed in polynomial time.
Moreover, the~$X$-tree~$\Tree'$ has a height of~3.

Before showing the correctness, we first prove that~$\Tree'$ is ultrametric.
In the rest of the proof we denote with~$\lambda_{w}$ the value~$\lambda(e)$, where~$e$ is the only edge that is directed towards~$w$. To show that~$\Tree'$ is ultrametric, for each~$x\in X\cup X^*$ we show that the paths from~$v$ to~$x$ have the same length as the edge from~$v$ to~$x_1$. This is sufficient because every path from the root to a taxon visits~$v$.
By definition, the claim is correct for every~$x_i\in X_1$.
For an~$x_i\in X_2$, the path from~$v$ to~$x_i$ is
\begin{eqnarray*}
	\lambda'_{u_i} + \lambda'_{x_i} = W \cdot (\lambda_{x_1} - \lambda_{x_i}) + W\cdot \lambda_{x_i} - \lambda_{x_1} = (W-1) \cdot \lambda_{x_1} = \lambda'_{x_1}
\end{eqnarray*}
By definition, this is also the length of the path from~$v$ to~$x_i^*\in X_i^*$.
We conclude that~$\Tree'$ is an ultrametric tree.

\pfpara{Correctness}
We now show that~\Instance is a yes-instance of~\unitc if and only if~$\Instance'$ is a yes-instance of~\unitc.

($\Rightarrow$)
Let~$S\subseteq X$ be a set of taxa.
Define~$S' := S \cup X^*$.
We show that~$S$ is a solution for~\Instance if and only if~$S'$ is a solution of~$\Instance'$.
First,~$|S'|=|S|+|X^*|$ and hence~$|S|\le B$ if and only if~$|S'|\le B'=B+|X^*|$. We compute the phylogenetic diversity of~$S'$ in~$\Tree'$. Here, we first consider the diversity from the subtrees below~$v$ and then the diversity from the edge~$(r,v)$.
For the subtrees rooted at~$v$ containing a leaf of~$X_1$ the total contribution is
$$D_1 = \sum_{x_i\in X_1 \cap S} \lambda'_{x_1}\cdot w(x_i)=\sum_{x_i\in X_1 \cap S} (W - 1) \cdot \lambda_{x_1}\cdot w(x_i).$$

For the subtrees rooted at~$v$ where~$S'$ contains~$x_i^*$ but not~$x_i$, the total contribution is
\begin{eqnarray*}
	D_2&=& \sum_{x_i\in X_2 \setminus S} (\lambda'_{u_i} + \lambda'_{x_i})\cdot w(x_i^*)\\
	&=& \sum_{x_i\in X_2 \setminus S} (W \cdot (\lambda_{x_1} - \lambda_{x_i}) + W\cdot \lambda_{x_i} - \lambda_{x_1})\cdot (1-W^{-1})\\
	&=& \sum_{x_i\in X_2 \setminus S} (W-1) \cdot \lambda_{x_1}\cdot (1-W^{-1}).
\end{eqnarray*}
For the subtrees rooted at~$v$ where~$S'$ contains~$x_i^*$ and~$x_i^*$, the total contribution is
\begin{eqnarray*}
	D_3&=& \sum_{x_i\in X_2 \cap S} \lambda'_{u_i}\cdot (1-(1-w(x_i^*))\cdot (1-w(x_i)))\\
	&& +\sum_{x_i\in X_2 \cap S} \lambda'_{x_i}\cdot w(x_i) + \sum_{x_i\in X_2 \cap S} \lambda'_{x_i}\cdot w(x_i^*)\\
	&=& \sum_{x_i\in X_2 \cap S} W \cdot (\lambda_{x_1} - \lambda_{x_i})\cdot (1-W^{-1}\cdot (1-w(x_i)))\\
	&& +\sum_{x_i\in X_2 \cap S} (W\cdot \lambda_{x_i} - \lambda_{x_1})\cdot (w(x_i)+1-W^{-1})\\
	&=& \sum_{x_i\in X_2 \cap S} (\lambda_{x_1} - \lambda_{x_i})\cdot (W - (1-w(x_i)))\\
	&& +\sum_{x_i\in X_2 \cap S} (W\cdot \lambda_{x_i} - \lambda_{x_1})\cdot (w(x_i)+1-W^{-1})\\
	&=& \sum_{x_i\in X_2 \cap S} \left[ \lambda_{x_1}\cdot (W - (1-w(x_i))-w(x_i)-1+W^{-1}))\right.\\
	&& \left. + \lambda_{x_i}\cdot (-W + (1-w(x_i))+W (w(x_i)+1-W^{-1}))\right] \\
	&=& \sum_{x_i\in X_2 \cap S} \left[ \lambda_{x_1}\cdot (W-2+W^{-1}) + \lambda_{x_i}\cdot (W-1) \cdot w(x_i)\right]\\
	&=& \sum_{x_i\in X_2 \cap S} (W-1)\cdot \left[ \lambda_{x_1}\cdot (1-W^{-1}) + \lambda_{x_i}\cdot w(x_i)\right].
\end{eqnarray*}
Finally, for the edge~$(r,v)$ the contribution is
\begin{eqnarray*}
	D_0&=& \lambda'_v\cdot \left( 1- \prod_{x_i^*\in X_2} (1-w(x_i^*))\cdot \prod_{x_i\in S} (1-w(x_i)) \right)\\
	&=& \lambda_v \cdot W^{|X_2|} \cdot (W - 1)\cdot \left( 1- W^{-|X_2|}\cdot \prod_{x_i\in S} (1-w(x_i))\right)\\
	&=& \lambda_v \cdot W^{|X_2|} \cdot (W - 1) - \lambda_v \cdot (W - 1)\cdot \prod_{x_i\in S} (1-w(x_i)).
\end{eqnarray*}
Altogether we conclude
\begin{eqnarray*}
	&& PD_{\Tree'}(S')\\
	&=& D_0 + D_1 + D_2 + D_3\\ % \lambda'_v\cdot \prod_{x_i^*\in X^*} (1-w(x_i^*))\cdot \prod_{x_i\in S} (1-w(x_i))\\
	% && +\sum_{x_i\in X_1 \cap S} \lambda'_{x_1}\cdot w(x_i)\\
	% && +\sum_{x_i\in X_2 \cap S} \lambda'_{u_i}\cdot (1-(1-w(x_i^*))\cdot (1-w(x_i)))\\
	% && +\sum_{x_i\in X_2 \cap S} \lambda'_{x_i}\cdot w(x_i) + \sum_{x_i\in X_2 \cap S} \lambda'_{x_1}\cdot (1-W)\\
	% && +\sum_{x_i\in X_2 \setminus S} (\lambda'_{u_i} + \lambda'_{x_i})\cdot w(x_i^*)\\
	&=& \lambda_v \cdot W^{|X_2|} \cdot (W - 1) - \lambda_v \cdot (W - 1)\cdot \prod_{x_i\in S} (1-w(x_i))\\
	&& +\sum_{x_i\in X_1 \cap S} (W - 1) \cdot \lambda_{x_1}\cdot w(x_i)\\
	&& +\sum_{x_i\in X_2 \setminus S} (W-1) \cdot \lambda_{x_1}\cdot (1-W^{-1})\\
	&& +\sum_{x_i\in X_2 \cap S} (W-1)\cdot ( \lambda_{x_1}\cdot (1-W^{-1}) + \lambda_{x_i}\cdot w(x_i))\\
        % (W-1)\cdot \sum_{x_i\in X_2 \cap S} \lambda_{x_1}\cdot (1-W^{-1}) + \lambda_{x_i}\cdot w(x_i)
	&=& (W-1)\cdot \left[ \lambda_v \cdot W^{|X_2|} - \lambda_v \cdot \prod_{x_i\in S} (1-w(x_i))\right.\\
	&& \left.+\sum_{x_i\in S} \lambda_{x_i}\cdot w(x_i) + \sum_{x_i\in X_2} \lambda_{x_1}\cdot (1-W^{-1})\right]\\
	&=& (W-1)\cdot \left[ \lambda_v \cdot \left(1- \prod_{x_i\in S} (1-w(x_i))\right) + \sum_{x_i\in S} \lambda_{x_i}\cdot w(x_i) \right.\\
	&& \left.+\lambda_v \cdot (W^{|X_2|}-1) + \sum_{x_i\in X_2} \lambda_{x_1}\cdot (1-W^{-1})\right]\\
	&=& (W-1)\cdot \left[ PD_\Tree(S) + (W^{|X_2|}-1)\cdot \lambda_v + |X_2|\cdot (1-W^{-1})\cdot \lambda_{x_1}\right]
\end{eqnarray*}
It directly follows that~$PD_{\Tree'}(S')\ge D'$ if and only if~$PD_{\Tree}(S)\ge D$.

($\Leftarrow$)
We show that~$\Instance'$ has a solution~$S'$ with~$X^*\subseteq S'$. Then,~$S'\setminus X^*$ is a solution for~\Instance as shown in the proof of the converse direction. 

Let~$S'$ be a solution for~$\Instance'$ that contains a maximum number of elements of~$X^*$ among all solutions.
If~$X^*\subseteq S'$, then we are done.
Otherwise, choose some~$x_i^* \not\in S'$. We show that there is a solution containing~$x_i^*$ and all elements of~$S'\cap X^*$, contradicting the choice of~$S'$.
If~$x_i\in S'$, then consider the set~$S_1 := (S'\setminus\{x_i\}) \cup \{x_i^*\}$. Now,~$|S_1|=|S'|\le B'$
and because~$w(x_i^*)=1-W^{-1}>w(x_i)$, we may conclude that~$PD_{\Tree'}(S_1)>PD_{\Tree'}(S')\ge D'$.
If~$S'\subset X^*$, then consider the set~$S_2 := X^*$. Now,~$|S_2|=|X^*|\le B'$ and~$PD_{\Tree'}(S_2)\ge PD_{\Tree'}(S')\ge D'$.
Finally, assume that~$x_i,x_i^*\not\in S'$ and~$x_j\in S'$ for some~$j\ne i$.
Again,~$w(x_i^*)=1-W^{-1}>w(x_i)$ and we know that the length of the path from~$v$ to~$x_i^*$ and~$x_j$ is the same. Consider the set~$S_3 := (S'\setminus\{x_j\}) \cup \{x_i^*\}$ and observe~$|S_3|=|S'|\le B'$. Moreover, by the above, \mbox{$PD_{\Tree'}(S_3)>PD_{\Tree'}(S')\ge D'$}.
\end{proof}

\section{Discussion}
In this paper, we studied the \GNAPLong (\GNAP), a natural generalization of the classical problem \NAP[0]{c_i}{1}{}.
We established several tractability and intractability results for \GNAP{} and some of its special cases.
Specifically, we showed that \GNAP is \W{1}-hard with respect to the number of taxa,~$|X|$, but can be solved in polynomial time when there are only a constant number of costs and \sprobs of projects.
We further introduced \unitc, a problem where all taxa are saved at the same cost, but may differ in their \sprob.
We showed that \unitc is \NP-hard, even on ultrametric trees of height~three.

Naturally, several open questions remain.
For example, it is not known whether \GNAP~admits a pseudopolynomial-time algorithm or whether \GNAP is strongly NP-hard.
Moreover, it remains open whether \GNAP is \FPT with respect to~$\var_c+\var_w$
or the budget~$B$. The latter is unresolved even for \unitc.

While projects, modeled through costs and \sprobs, are one way to capture a better decision-making process in conservation planning, other approaches have also been studied.
One direction incorporates ecological dependencies, such as a food webs which represent predator-prey relationships.
A species can then only be preserved if it is either a source of the food web, or can find sufficient prey among other preserved species~\cite{moulton,WeightedFW}. So far, the complexity of the problems has been analyzed only in the simple setting where~$k$ taxa whose survival probability is lifted from 0 to 1 are selected~\cite{moulton,KS24,WeightedFW}.    
Another line of research generalizes phylogenetic trees to networks, allowing \emph{horizontal gene transfer} and \emph{hybridization} in so called \emph{reticulations}.
This concept, however, requires new definitions for phylogenetic diversity on networks.
Over the years, several such definitions have been introduced~\cite{WickeFischer2018,bordewichNetworks,van2025average}.
More recently, these two perspectives---ecological dependencies and phylogenetic networks---have been combined, analyzed, and efficient algorithms have been presented~\cite{MAPPD-Dependencies}.
It seems only natural to study the complexity of \GNAP and its various special cases when they are combined with these richer models of phylogenetic diversity.

\bibliographystyle{elsarticle-num} 
\bibliography{References}

\end{document}